%% file: ms.tex
\documentclass[a4paper,orivec,envcountsame]{llncs}
\usepackage{verbatim} 

 \usepackage{enumerate}  
 \usepackage{enumitem}
\usepackage{todonotes}
\usepackage{mathtools}
\usepackage{silence}
\usepackage{amsmath}
\usepackage{thm-restate}
\usepackage{microtype}
\usepackage{accsupp}
\usepackage{thm-restate}
\usepackage{amssymb}
\usepackage{stmaryrd}

\newcommand{\dltrans}[1]{{\sf dl}(#1)}
\newcommand{\variables}[1]{\ensuremath{{\sf Var}(#1)}\xspace}

\newcommand{\valueleft}{\ensuremath{v}\xspace}
\newcommand{\valueright}{\ensuremath{w}\xspace}

\newcommand{\dt}{\ensuremath{{\sf time}}\xspace}

\newcommand{\ta}{\ensuremath{\mathbb{T}\xspace}}
\newcommand{\bef}{\ensuremath{{\sf before}}\xspace} 
\newcommand{\aft}{\ensuremath{{\sf after}}\xspace}
\newcommand{\since}{\ensuremath{{\sf since}}\xspace}

\renewcommand{\between}{\ensuremath{{\sf between}}\xspace}
\newcommand{\unt}{\ensuremath{{\sf until}}\xspace}

\newcommand{\during}{\ensuremath{{\sf during}}\xspace}

\newcommand{\intervals}[1]{\ensuremath{\mathbb{N}^{#1}_\akb}\xspace}

\usepackage{mathtools}

\newcommand{\valtype}{\mathsf{valtype}}
\newcommand{\setvariables}{\NU}
\newcommand{\objectvariables}{\NV}

\newcommand{\linearmap}{\ensuremath{f}\xspace}

\newsavebox{\spacebox}
\begin{lrbox}{\spacebox} 
\verb*! !
\end{lrbox}

\input{macros}
\newcommand{\ind}{\ensuremath{{\sf i}}\xspace}
\newcommand{\ann}{\ensuremath{{\sf a}}\xspace}
\newcommand{\dlliterhorn}{DL-Lite$^\Hmc_\mathit{horn}$\xspace}
\newcommand{\dllitea}{DL-Lite$^{\Hmc,@}_{\mathit{horn}}$\xspace}
\newcommand{\dlliteat}{DL-Lite$^{\Hmc,\mathbb{T},@}_{\mathit{horn}}$\xspace}
\renewcommand{\implies}{\ensuremath{\Rightarrow}\xspace}

\begin{document}

\title{
Geometric Models for (Temporally) Attributed Description Logics
} 
\author{Camille Bourgaux$^1$, Ana Ozaki$^2$, Jeff Z. Pan$^3$}
\institute{DI ENS, ENS, CNRS, PSL University \& Inria, Paris, France \and University of Bergen, Norway \and University of Edinburgh, United Kingdom}

\maketitle
\begin{abstract}
In the search for knowledge graph embeddings that could capture ontological knowledge, geometric models of existential rules have been recently introduced. 
It has been shown that   convex geometric regions  capture the so-called quasi-chained rules. 
Attributed description logics (DL) have been defined to bridge the gap between DL languages and knowledge graphs, whose facts often come with various kinds of annotations that may need to be taken into account for reasoning. 
In particular, temporally attributed DLs are enriched by specific attributes whose semantics allows for some temporal reasoning. 
Considering that geometric models and (temporally) attributed DLs are promising tools designed for knowledge graphs, this paper investigates their compatibility, focusing on the attributed version of a Horn dialect of the DL-Lite family. 
We first adapt the definition of geometric models to attributed DLs and show that every satisfiable ontology has a convex geometric model. 
Our second contribution is a study of the impact of temporal attributes. We show that a temporally attributed DL may not have a convex geometric model in general but we can recover geometric satisfiability by imposing some restrictions on the use of the temporal attributes.
\end{abstract}
\section{Introduction}

Knowledge graph embeddings 
are 
popular latent representations of knowledge graphs (KG). In the search for KG embeddings that could capture ontological knowledge (i.e., schema of KG), 
geometric models of existential rules have been recently introduced~\cite{DBLP:conf/kr/Gutierrez-Basulto18}. 
Such models have several advantages. Notably, they ensure that facts which are valid in the embedding are logically consistent and deductively closed w.r.t. the ontology, and they can also be used to find plausible missing ontology rules. 
It has been shown that   convex geometric regions  capture the so-called quasi-chained rules, a fragment 
of first-order Horn logic. 
Attributed description logics (DL) have been defined to bridge the gap between DL ontology languages and KG, 
whose facts often come with various kinds of annotations that may need to be taken into account for reasoning. 
In particular, they were introduced as a formalism for dealing with the 
 meta-knowledge present in KG, such as temporal validity, provenance, language, and others~\cite{KMOT2017:ADLs,DBLP:conf/ijcai/Krotzsch0OT18,bourgauxozaki}. 
As time is of primary interest in KG, attributed DLs have been enriched with temporal attributes, 
whose semantics allows for some temporal reasoning over discrete time \cite{DBLP:conf/birthday/OzakiKR19}. 
Considering that geometric models and (temporally) attributed DLs are promising tools designed for KG, this paper investigates their compatibility, focusing on the attributed version of a Horn dialect of the DL-Lite family. 

Our contributions are as follows:
\begin{itemize}
\item We adapt the notion of geometric models for ((temporally) attributed) DLs; in particular, we use an arbitrary linear map to combine the individual geometric interpretations instead of restricting ourselves to vector concatenation, and define satisfaction of concept or role inclusions directly based on geometric inclusion relationship between the regions that interpret  the concepts or roles. 
\item We show that every satisfiable attributed \dlliterhorn ontology has a convex geometric model but
there are  satisfiable \emph{temporally} attributed \dlliterhorn ontologies without such a model.
\item We exhibit restrictions on the use of temporal attributes that guarantee 
that temporally attributed 
\dlliterhorn ontologies 
have a convex geometric model.
\end{itemize}

We define attributed DLs and geometric models in Section~\ref{sec:geom-models}.
Then, in Section~\ref{sec:satisfiability}, we study the relationship between satisfiability 
and the existence of convex geometric models in \dlliterhorn.
We then extend our analysis for temporally attributed \dlliterhorn in Section~\ref{sec:addingtime}.
In Section~\ref{sec:relatedwork}, we discuss related works
and we conclude in Section~\ref{sec:conclusion}. 
Omitted proofs are given in the appendix. 

\section{Geometric Models for Attributed Description Logics}\label{sec:geom-models}

In this section, we recall the framework of attributed DLs and define geometric models in this context.

\subsection{Attributed DLs}\label{sec:attributedDLs}

We introduce attributed DLs by defining 
attributed DL-Lite~\cite{bourgauxozaki}, focusing on the \dlliterhorn dialect. The notions presented here  can be easily
adapted to other attributed DLs, e.g. \EL, as in~\cite{KMOT2017:ADLs,DBLP:conf/ijcai/Krotzsch0OT18}. 
Let \NC, \NR, and $\NI$ be countably infinite and mutually disjoint sets 
of   \emph{concept}, \emph{role},  and \emph{individual names}. 
We assume that \NI is divided into two sets, called $\dlNames{i}$ and $\dlNames{a}$,
and we  refer to the elements in $\dlNames{a}$ as \emph{annotation names}. 
  We consider   
  an additional set \setvariables of \emph{set variables}
  and a set \objectvariables of \emph{object variables}. 
  The set $\Slang$ of \emph{specifiers} contains the following expressions:
  \begin{itemize}
  \item set variables $X\in\setvariables$;
  \item \emph{closed specifiers} $\closedAnnotation{a_{1} \syntaxequality v_{1}, \dotsc, a_{n} \syntaxequality v_{n}}$; and
  \item \emph{open specifiers} $\openAnnotation{a_{1} \syntaxequality v_{1}, \dotsc, a_{n} \syntaxequality v_{n}}$,
  \end{itemize}
  where $a_i \in \dlNames{a}$  and $v_i$ is either  an individual name in $\dlNames{a}$,
  an object variable in \objectvariables,
  or an expression of the form 
  $\syntaxintent{X}{a}$, with $X$ a set variable in $\setvariables$ and $a$ an
  individual name in $\dlNames{a}$. 
  We use  $\syntaxintent{X}{a}$ to refer  to the (finite, possibly empty) set of
  all values of attribute $a$ in an annotation set $X$.
  A \emph{ground specifier} is a closed or open specifier built only over $\dlNames{a}$.

  \mypar{Syntax.}
 A  \dllitea  \emph{concept} (resp. \emph{role}) \emph{assertion} is an expression $\syntaxpatom{A(a)}{S}$
 (resp. $\syntaxpatom{ R(a,b)}{S}$), with $A\in\NC$ (resp. $R\in\NR$), $a,b\in\dlNames{i}$, and $S\in\Slang$ a
 ground closed specifier. 
  A  \dllitea  \emph{role inclusion} is an expression of the form:
 \begin{align}
 \syntaxsetrestriction{X}{S}  \quad (P\sqsubseteq Q),\label{eq_gci}
 \end{align}
 where  $S\in\Slang$ is a 
 closed or open specifier, $X \in\setvariables$ is a set variable, 
  and $P,Q$ are role expressions built according to the following syntax:
  \begin{align}
  P \coloneqq \syntaxpatom{R}{S}\mid \syntaxpatom{R^-}{S}, \quad\quad  Q \coloneqq P \mid \neg P
 \label{eq_dllite_roles}
 \end{align}
 with $S \in \Slang$, $R \in \NR$. 
A  \dllitea  \emph{concept inclusion} is 
of the form:
 \begin{align}
 \syntaxsetrestriction{X_1}{S_1}, \ldots, \syntaxsetrestriction{X_n}{S_n} \quad  (\bigsqcap_{i=1}^k B_i\sqsubseteq C),\label{eq_gci}
 \end{align}
 where $k,n\geq 1$,  $S_1, \dots, S_n\in\Slang$ are
 closed or open specifiers, $X_1,\dots, X_n\in\setvariables$ are set variables, 
 and $B_i,C$ are concept expressions built according to:
  \begin{align}
  B\coloneqq\syntaxpatom{A}{S} \mid \exists P, \quad 
   C \coloneqq B \mid \bot,  
 \label{eq_dllite_concepts}
 \end{align}
 where $P$ is as in Equation~\ref{eq_dllite_roles}, $A \in \NC$ and $S \in \Slang$. 
 Role expressions of the form $P$ are called \emph{roles} 
 and concept expression of the form $B$ are \emph{basic concepts}. 
 We further require that all object variables are \emph{safe}, that is, 
 if they occur on the right side 
 of a concept/role inclusion or in a specifier associated with a set variable
 occurring on the right side then they  must also occur on the left side of the inclusion
 (or in a specifier associated with a set variable occurring on the left).
 
 A \dllitea  \emph{ontology}  is a set of \dllitea assertions, role and concept inclusions.
 We say that an inclusion is \emph{positive} if it does not contain negation or $\bot$.  
 Also, we say that a \dllitea ontology is \emph{ground} if it does not contain   variables. 
 To simplify notation, we omit the specifier $\openAnnotation{}$ (meaning ``any annotation set'')
 in role or concept expressions.
 In this sense, any \dlliterhorn axiom is also a \dllitea  axiom.
 Moreover, we omit prefixes of the form $\syntaxsetrestriction{X}{\openAnnotation{\ }}$, which
   state that there is no restriction on $X$.

\mypar{Semantics.} An \emph{interpretation} $\Inter = \tuple{\Delta^\Inter_{\ind},\Delta^\Inter_{\ann},\cdot^\Inter}$ of
an attributed DL consists of a non-empty domain $\Delta^\Inter_{\ind}$ of
individuals, a non-empty domain $\Delta^\Inter_{\ann}$ of
annotations,
and a function $\cdot^\Inter$. 
Individual names $a\in\dlNames{i}$ are interpreted as elements $a^\Inter\in\Delta^\Inter_{\ind}$ and individual names $a\in\dlNames{a}$ are interpreted as elements $a^\Inter\in\Delta^\Inter_{\ann}$.
To interpret annotation sets, we use the set
$\asdom{\Inter}\defeq\set{\Sigma\subseteq\Delta^\Inter_{\ann}\times \Delta^\Inter_{\ann}\mid\Sigma \text{ is finite}}$ of all
finite binary relations over $\Delta^\Inter_{\ann}$.
Each concept name $A\in\NC$ is interpreted as a set $A^\Inter\subseteq \Delta^\Inter_{\ind}\times\asdom{\Inter}$ of elements with annotations, and each role name
$ R\in \NR$ is interpreted as a set $R^\Inter\subseteq \Delta^\Inter_{\ind}\times\Delta^\Inter_{\ind}\times\asdom{\Inter}$ of pairs of elements with annotations.
Each element (pair of elements) may appear with multiple different
annotations.
$\Inter$ \emph{satisfies} a concept assertion $\syntaxpatom{A(a)}{\closedAnnotation{a_1\syntaxequality v_1, \dotsc, a_n\syntaxequality v_n}}$ if
$\tuple{a^\Inter,\{\tuple{a_1^\Inter,v_1^\Inter},\ldots,\tuple{a_n^\Inter,v_n^\Inter}\}}\in A^\Inter$. 
Role assertions are interpreted analogously.
Expressions with free set or object variables are interpreted using variable assignments
$\Zuweisung$ mapping object variables $x \in \objectvariables$ to elements
$\Zuweisung(x) \in \Delta^\Imc_{\ann}$
and set variables  $X\in \setvariables$ to finite binary relations $\Zuweisung(X) \in \asdom{\Inter}$. 
For convenience, we also extend variable assignments to individual names, setting $\Zuweisung(a)=a^\Inter$ for every $a\in\dlNames{a}$. 
A specifier $S\in\Slang$ is interpreted as a set $S^{\Inter,\Zuweisung}\subseteq\asdom{\Inter}$
of matching annotation sets. We set $X^{\Inter,\Zuweisung}\defeq\{\Zuweisung(X)\}$ 
for variables $X\in\setvariables$. The semantics of closed specifiers is defined as: 
\begin{itemize} 
\item $\closedAnnotation{a\syntaxequality v}^{\Inter,\Zuweisung}\defeq\{\{\tuple{a^\Inter, \Zuweisung(v)}\}\}$ 
where $v \in \dlNames{a}\cup\NV$;
\item $\closedAnnotation{a\syntaxequality X.b}^{\Inter,\Zuweisung}\defeq
\{\{\tuple{a^\Inter,\delta} \mid  \tuple{b^\Inter,\delta}\in\Zuweisung(X)\}\}$;
\item $\closedAnnotation{a_1\syntaxequality v_1, \dotsc, a_n\syntaxequality v_n}^{\Inter,\Zuweisung}\defeq \{\bigcup_{i=1}^n F_i\}$ where $\{F_i\}=  \closedAnnotation{a_i\syntaxequality v_i}^{\Inter,\Zuweisung}$ for all $1\leq i\leq n$.
\end{itemize}
$S^{\Inter,\Zuweisung}$ therefore is a singleton set for variables and closed specifiers. For open specifiers, however, we define
$\openAnnotation{a_1\syntaxequality v_1, \dotsc, a_n\syntaxequality v_n}^{\Inter,\Zuweisung}$ to be the set:
\[\{ F\subseteq\asdom{\Inter}\mid F\supseteq G\text{ for }\{G\}=\closedAnnotation{a_1\syntaxequality v_1, \dotsc, a_n\syntaxequality v_n}^{\Inter,\Zuweisung}\}.\]
Now given $A\in\NC$, $ R\in \NR$, and $S\in\Slang$, we define:
\begin{align*}
(\syntaxpatom{A}{S})^{\Inter,\Zuweisung}
	& \defeq \{\delta \mid \tuple{\delta,F}\in A^\Inter \text{ for some }F\in S^{\Inter,\Zuweisung}\}, \\
(\syntaxpatom{R}{S})^{\Inter,\Zuweisung}
	& \defeq \{\tuple{\delta,\epsilon} \mid \tuple{\delta,\epsilon,F}\in R^\Inter \text{ for some } F\in S^{\Inter,\Zuweisung}\}.
\end{align*}
Further DL expressions are defined as: 
$(\syntaxpatom{R^-}{S})^{\Inter,\Zuweisung} \defeq\{(\gamma,\delta)\mid  \tuple{\delta,\gamma}\in (\syntaxpatom{R}{S})^{\Inter,\Zuweisung}\}$, 
$\neg P^{\Inter,\Zuweisung} \defeq(\Delta^\Inter_{\ind}\times\Delta^\Inter_{\ind})\setminus P^{\Inter,\Zuweisung}$, 
$\exists P^{\Inter,\Zuweisung} \defeq\{\delta\mid \text{there is }\tuple{\delta,\epsilon}\in P^{\Inter,\Zuweisung} \}$,  $(B_1\sqcap B_2)^{\Inter,\Zuweisung}  \defeq B^{\Inter,\Zuweisung}_1 \cap B^{\Inter,\Zuweisung}_2$, 
$ \bot^{\Inter,\Zuweisung} \defeq\emptyset$. 
$\Inter$ \emph{satisfies} a concept inclusion of the form \eqref{eq_gci} if,
for all variable assignments $\Zuweisung$ that satisfy
$\Zuweisung(X_i)\in S^{\Inter,\Zuweisung}_i$ for all $1\leq i\leq n$,
we have $(\bigsqcap_{i=1}^k B_i)^{\Inter,\Zuweisung}\subseteq C^{\Inter,\Zuweisung}$.
Satisfaction of role inclusions is defined analogously.
An interpretation $\Inter$ satisfies an ontology $\Omc$, or is a \emph{model} of $\Omc$, if it satisfies all of its axioms.
As usual, $\models$ denotes 
the induced
logical entailment relation.

For ground specifiers $\{S,T\}\subseteq \Slang$, 
we write $S\implies T$ if 
$T$ is an 
open specifier, and the set of attribute-value pairs $a : b$ in
$S$ is a superset of the set of attribute-value pairs in $T$.

\subsection{Geometric Models}

We now define the geometric interpretations of attributed relations. 
Let $m$ be an integer and $\linearmap : \mathbb{R}^m\times\mathbb{R}^m\mapsto \mathbb{R}^{2\cdot m}$ 
be a fixed but arbitrary  linear map satisfying the following: 
\begin{enumerate}[label=\textit{(\roman*)}]
\item \label{item1} the restriction of $\linearmap$ to $\mathbb{R}^m\times \{0\}^{m}$ is injective;
\item \label{item2} the restriction of $\linearmap$ to $ \{0\}^{m}\times \mathbb{R}^m$ is injective;
\item  \label{item3} $\linearmap(\mathbb{R}^m\times \{0\}^m)\cap \linearmap(\{0\}^m\times \mathbb{R}^m)=\{0^{2\cdot m}\}$;
\end{enumerate} 
where $0^{m}$ denotes the vector $(0,\ldots,0)$ with $m$ zeros. 
Intuitively, individuals will be interpreted as vectors from $\mathbb{R}^m$ and $\linearmap$ will be used to combine two vectors to interpret pairs of individuals.

\begin{definition}[Geometric Interpretation]\label{def:sat-assertions}
An $m$-dimensional $\linearmap$-geometric interpretation $\eta$ of $(\NC,\NR,\dlNames{i},\dlNames{a} )$ assigns 
\begin{itemize}
\item to each $A\in\NC$ and  ground $S\in\Slang$ a region $\eta(\syntaxpatom{A}{S})\subseteq \mathbb{R}^{m}$,
\item to each $R\in \NR$ and  ground $S\in\Slang$ a region $\eta(\syntaxpatom{R}{S})\subseteq \mathbb{R}^{2\cdot m}$, 
and 
\item to each $a\in \dlNames{i}$ a vector $\eta(a)\in\mathbb{R}^m$. 
\end{itemize}
Moreover, for all 
$\{S,T\}\subseteq \Slang$ and $E\in\NC\cup\NR$, 
if $S\implies T$ then $\eta(\syntaxpatom{E}{S})\subseteq  \eta(\syntaxpatom{E}{T})$. 
We say that $\eta$ is \emph{convex} if, for every $E\in\NC\cup\NR$, every ground
$S\in\Slang$, every $\vec{v}_1,\vec{v}_2\in \eta(\syntaxpatom{E}{S})$, and every $\lambda\in [0,1]$, if $\vec{v}_1,\vec{v}_2\in \eta(\syntaxpatom{E}{S})$ then
$(1-\lambda)\vec{v}_1+\lambda \vec{v}_2\in \eta(\syntaxpatom{E}{S})$. 

The interpretation of ground complex concept or role expressions is as follows. 
Assume all specifiers occurring in expressions below are ground, $R\in\NR$,  $P$ is a role, and $B, B_i$ are basic concepts. 
Then,
\begin{itemize}
\item $\eta(\syntaxpatom{R^-}{S})\defeq\{f(\delta,\delta')\mid f(\delta',\delta)\in\eta(\syntaxpatom{R}{S})\}$,
\item $\eta(\neg P)\defeq\mathbb{R}^{2\cdot m} \setminus \eta(P)$,
\item $\eta(\exists P)\defeq\{\delta\mid \exists \delta',\linearmap(\delta,\delta')\in \eta(P) \}$,
\item $\eta(\bigsqcap_{i=1}^n B_i)\defeq\bigcap_{i=1}^n\eta(B_i)$,
\item $\eta(\bot)\defeq \emptyset$.
\end{itemize}
\end{definition}
We may omit ``of $(\NC,\NR,\dlNames{i},\dlNames{a} )$''
when we speak about $m$-dimensional $\linearmap$-geometric interpretations. The interest of geometric interpretations is that concept and role assertions translate into membership in geometric regions and ground concept or role inclusions translate into geometric inclusions.

\begin{definition}[Satisfaction of Ground Axioms]
An $m$-dimensional $\linearmap$-geometric interpretation $\eta$ satisfies 
\begin{itemize}
\item a concept assertion $\syntaxpatom{A(a)}{S}$, denoted $\eta\models \syntaxpatom{A(a)}{S}$, if $\eta(a) \in \eta(\syntaxpatom{A}{S})$;
\item a role assertion $\syntaxpatom{R(a,b)}{S}$, denoted $\eta\models \syntaxpatom{R(a,b)}{S}$, if 
 $  \linearmap(\eta(a),\eta(b)) \in \eta(\syntaxpatom{R}{S})$; 
 \item a ground role inclusion $P\sqsubseteq Q$, denoted 
$\eta\models P\sqsubseteq Q$, if   
$\eta(P)\subseteq \eta(Q)$;
\item a ground concept inclusion $\bigsqcap_{i=1}^n B_i\sqsubseteq C$, denoted 
$\eta\models \bigsqcap_{i=1}^n B_i\sqsubseteq C$, if we have that   
$\eta(\bigsqcap_{i=1}^n B_i)\subseteq \eta(C)$. 
\end{itemize}
\end{definition}

We are ready for the first theorem which establishes that our more general
notion of geometric models still has the same properties of the geometric models
 originally proposed~\cite{DBLP:conf/kr/Gutierrez-Basulto18}.

\begin{theorem}\label{thm:linearmapTh}
Let  $\eta$ be an $m$-dimensional $f$-geometric interpretation. For every linear map $f'$ satisfying (i)-(iii),
the $m$-dimensional $f'$-geometric interpretation $\eta'$ defined as:
\begin{itemize}
\item $\eta'(a):=\eta(a)$, for all $a\in \dlNames{i}$;
\item $\eta'(\syntaxpatom{A}{S}):=\eta(\syntaxpatom{A}{S})$, for all $A\in\NC$ and  ground $S\in\Slang$; and 
\item $\eta'(\syntaxpatom{R}{S}):=\{f'(\delta,\delta')\mid f(\delta,\delta')\in\eta(\syntaxpatom{R}{S})\}$, for all $R\in\NR$ and  ground $S\in\Slang$;
\end{itemize}
is such that $\eta\models\alpha$ iff $\eta'\models\alpha$, for all ground axioms $\alpha$. 
\end{theorem}
\begin{proof}[Sketch] This result follows from the fact that there is an isomorphism 
between the regions in $\eta$ and $\eta'$. 
\end{proof}

To define when a geometric interpretation is a model of a (possibly not ground) \dllitea ontology, we need to define when such an interpretation satisfies non-ground concept or role inclusions. To do so, we use a standard interpretation built from the geometric interpretation. 
Given a ground specifier $S$ and an annotation name $\star$, we define $F^\star_S=\{(a,b)\mid a\syntaxequality b \text{ occurs in } S\}\cup\{(\star,\star)\mid \text{if }S\text{ is open}\}$. Given an $m$-dimensional $\linearmap$-geometric interpretation $\eta$, 
a subset $D_\ann$ of $\dlNames{a}$ and an annotation name $\star\in \dlNames{a}\setminus D_\ann$, 
we define an interpretation $\Imc(\eta, D^\star_\ann)$
as follows. 
\begin{itemize}
\item $\Delta^{\Imc(\eta, D^\star_\ann)}_\ind=\mathbb{R}^m$ and, for all $a\in\dlNames{i}$, $a^{\Imc(\eta, D^\star_\ann)}=\eta(a)$;
\item $\Delta^{\Imc(\eta, D^\star_\ann)}_\ann= \dlNames{a}$ and for all $a\in \dlNames{a}$, $a^{\Imc(\eta, D^\star_\ann)}=a$;
\item  $A^{\Imc(\eta, D^\star_\ann)}=\{(\delta,F^\star_S)\mid \delta \in\eta(\syntaxpatom{A}{S}), S \text{ built on }D_\ann\}$, for all $A\in\NC$;
\item $R^{\Imc(\eta, D^\star_\ann)}=\{(\delta, \epsilon ,F^\star_S)\mid  \linearmap(\delta, \epsilon) \in\eta(\syntaxpatom{R}{S}), S \text{ built on }D_\ann\}$, for all $R\in\NR$.
\end{itemize}

The next proposition shows that $\eta$ and $\Imc(\eta, D^\star_\ann)$ satisfy the same ground axioms built using only annotation names from $D_\ann$. It follows in particular that $\eta$ satisfies all ground axioms of an ontology $\Omc$ iff $\Imc(\eta, \NIkbstar{\Omc})$ does, where $\NIkb{\Omc}$ is the set of annotation names from $\dlNames{a}$ that occur in $\Omc$. 

\begin{restatable}{theorem}{lemfundam}
\label{lem:fundam}
Let $\eta$ be an $m$-dimensional $\linearmap$-geometric interpretation. 
Let $\alpha$ be a ground axiom, \ie $\alpha$ is either a concept/role assertion 
 or a ground concept/role inclusion. 
Let $D_\alpha$ be the set of annotation names that occur in $\alpha$ and let $D$ be a subset 
of $\dlNames{a}$ such that $D_\alpha \subseteq D$. Let $\star$ be an annotation name that does not occur in $D$. 
Then, the following holds: $\eta\models \alpha$ iff $ \Imc(\eta, D^\star)\models \alpha$. 
\end{restatable}
\begin{proof}[Sketch]
The proof relies heavily on the definition of $F^\star_S$ and the requirement that $\eta(\syntaxpatom{E}{S})\subseteq \eta(\syntaxpatom{E}{T})$ when $S\implies T$ in Definition \ref{def:sat-assertions}.
\end{proof}

We are now ready to define geometric models of \dllitea ontologies.
\begin{definition}[Geometric Model]
Let \Omc be a \dllitea ontology, and let $\NIkb{\Omc}$ be the set of annotation names from $\dlNames{a}$ that occur in $\Omc$ and $\star$ an annotation name that does not occur in \Omc.  An $m$-dimensional $\linearmap$-geometric interpretation $\eta$ is a model of \Omc 
if $\Imc(\eta, \NIkbstar{\Omc})$ is a model of \Omc. 
\end{definition}

\section{Satisfiability and Convex Geometric Models}\label{sec:satisfiability}

We start by recalling some definitions and 
results  on geometric interpretations of an ontology 
containing existential rules~\cite{DBLP:conf/kr/Gutierrez-Basulto18}. 
An existential rule is an expression of the form $B_1\wedge\dots\wedge B_n\rightarrow \exists X_1,\dots, X_j. H$ where $n\geq 0$, the $B_i$'s and $H$ are atoms built from sets of predicates, constants and variables, and the $X_i$'s are variables. A negative constraint is a rule whose head is $\bot$. An existential rule or negative constraint is \emph{quasi-chained} if for all $1\leq i\leq n$, $|(\mn{vars}(B_1)\cup\dots\cup\mn{vars}(B_{i-1}))\cap \mn{vars}(B_i)|\leq 1$, where $\mn{vars}(B)$ denotes the variables that occur in $B$. 
It is easy to see that a \dlliterhorn ontology without negative role inclusions can be translated into a  quasi-chained ontology. Negative role inclusions are not quasi-chained: their translation to rules is indeed of the form $P_1(x,y)\wedge P_2(x,y)\rightarrow \bot$ where the body atoms share two variables. 
A (standard) model $\Mmc$ of an existential rules ontology $\Kmc$ is a set of facts that contains all facts from $\Kmc$ and satisfies all existential rules from $\Kmc$. In this setting, for every fact $\alpha$, $\alpha\in \Mmc$ iff $\Mmc\models\alpha$. 

Given a set $\Rmc$ of relation names and a set $X$ of constants and labelled nulls, a $m$-dimensional geometric interpretation $\eta$ of $(\Rmc,X)$ assigns to each $k$-ary relation $R$ from $\Rmc$ a region $\eta(R)\subseteq \mathbb{R}^{k.m}$ and to each object $o$ from $X$ a vector $\eta(o)\in \mathbb{R}^m$. 
Tuples of individuals are interpreted using vectors concatenation, which plays the role of the linear map $\linearmap$ 
we use to interpret a pair of individuals:  
for every $R\in\Rmc$ and $o_1,\dots,o_k\in X$, $\eta\models R(o_1,\dots, o_k)$ if $\eta(o_1)\oplus \dots\oplus\eta(o_k)\in\eta(R)$. 
The authors define 
\[\Phi(\eta)=\{R(o_1,\dots,o_k)\mid R\in \Rmc, o_1,\dots, o_k\in X, \eta\models R(o_1,\dots, o_k)\}.\]

Proposition 3 in \cite{DBLP:conf/kr/Gutierrez-Basulto18} states that if $\Kmc$ is a quasi-chained ontology and $\Mmc$ is a finite model of $\Kmc$, then $\Kmc$ has a convex geometric model $\eta$ such that $\Phi(\eta)=\Mmc$. 
This transfers to the DL setting as follows. Let \Omc be a quasi-chained DL ontology   and $\Jmc$ be a finite 
model of $\Omc$ such that $\Delta^{\Jmc}=\dlNames{i}$, and $a^{\Jmc}=a$ for every $a\in\dlNames{i}$  
(which implies that for every concept $A$, if $\delta\in A^{\Jmc}$, then $\Jmc\models A(\delta)$, and similarly for roles).
Then \Omc has a convex geometric model $\eta$ such that 
$$\{E(\vec{t})\mid \eta\models E(\vec{t})\}=\{E(\vec{t})\mid \Jmc\models E(\vec{t})\},$$ 
where $E\in\NC\cup\NR$
and $\vec{t}$ is a tuple with the arity of $E$. In the following, we may similarly 
write $\syntaxpatom{E(\vec{t})}{S}$ to refer to an atom of   the form $\syntaxpatom{A(a)}{S}$
or $\syntaxpatom{R(a,b)}{S}$.

\begin{restatable}{theorem}{dlliteaconvexmodel}\label{dlliteaconvexmodel}
Let \Omc be a satisfiable \dllitea ontology 
without negative role inclusions. 
\Omc has a convex geometric model. 
\end{restatable}
\begin{proof}[Sketch]
We use grounding to translate \Omc into an equisatisfiable \dlliterhorn ontology. 
Since $\Omc$ does not contain negative role inclusion, the obtained \dlliterhorn ontology is quasi-chained. 
We can thus apply the result by Gutiérrez-Basulto and Schockaert~\cite{DBLP:conf/kr/Gutierrez-Basulto18} to get a convex $m$-dimensional $\linearmap$-geometric model with $\linearmap$ being vector concatenation. This geometric model is used to construct a geometric model of \Omc.
\end{proof}

\section{Adding Time}\label{sec:addingtime}

In this section, we discuss the ability of convex geometric models to capture temporally attributed DLs. 
We show that we need to restrict the expressivity of the temporal ontology to get a convex geometric model. 

We introduce temporally attributed DLs by defining temporally attributed \dlliterhorn, called \dlliteat, 
as in \cite{DBLP:conf/birthday/OzakiKR19}. 
The description logic \dlliteat is defined as a multi-sorted version of \dllitea,
where time points and intervals are seen as  datatypes.
 \emph{Time points} are elements of \NTimepoints, and \emph{time intervals} are elements of 
 \NIntervals. 
 These sets and the set of \emph{(abstract) individual names} \NI are  mutually disjoint.
Time points are represented in a discrete manner by natural numbers, and we assume that
elements of $\NTimepoints$ ($\NIntervals$) are (pairs of) numbers. 
A pair of numbers $k,\ell$ in $\NIntervals$ is denoted $[k,\ell]$. 

The  annotation 
names: $\dt,\bef,\aft,\unt,\since,\during,\between\in \dlNames{a}$  are called \emph{temporal attributes} and have their own semantics. 
Basically, \dt is used to mark a point in time and \bef and \aft 
refer to \emph{some} point in the past and in the future, respectively. The temporal attributes \unt and \since
refer to \emph{all} points in the past and all point in the future (e.g. \since $2020$ the KR conference became an annual event). 
Finally, \during is an interval which represents
a period of time (it refers to all points in the interval) and  \between is an interval 
of uncertainty for when an event happened.
The \emph{value type} of $\dt,\bef,\aft,\unt,\since$ is \NTimepoints, while 
the value type of $\during,\between$ is \NIntervals. 
We write ${\sf valtype}(a)$ to refer to the value type of 
the annotation name $a$. Object variables  are now taken from pairwise disjoint sets ${\sf Var}(\dlNames{a})$,
${\sf Var}(\NTimepoints)$, and ${\sf Var}(\NIntervals)$.

Annotation set specifiers are defined as in Section~\ref{sec:attributedDLs} with 
the difference that for each
$a_i\in\dlNames{a}$ and each $v_i$ in attribute value pair $a_i:v_i$ 
we require compatibility between the value type of its attribute, that is:
\begin{itemize}
\item $v_i\in{\sf valtype}(a_i)\cup{\sf Var}({\sf valtype}(a_i)) $, 
or  
\item $v_i=[\valueleft,\valueright]$ with ${\sf valtype}(a_i)=\NIntervals$
and $\valueleft,\valueright$    in $\NTimepoints\cup{\sf Var}(\NTimepoints)$, or 
\item
$v_i=X.b$ with $X\in\NU$, $b\in\dlNames{a}$, and ${\sf valtype}(a_i)={\sf valtype}(b)$. 
\end{itemize}

A \emph{time-sorted interpretation} $\Inter=\tuple{\Delta^\Inter_{\sf i},\Delta^\Inter_{\sf a},\cdot^\Inter}$ 
is an interpretation with a domain $\Delta^\Inter_{\sf a}$ that is a disjoint union of 
$\Delta^\Inter_A\cup\Delta^\Inter_T\cup\Delta^\Inter_{2T}$, where $\Delta^\Inter_A$ 
is the \emph{abstract domain of annotations},
$\Delta^\Inter_T$ (the \emph{temporal domain})  is 
a finite 
or infinite interval, 
and $\Delta^\Inter_{2T}=\Delta^\Inter_T\times\Delta^\Inter_T$.
We interpret individual names in $\dlNames{i}$ as elements in $\Delta^\Inter_{\sf i}$;
annotation names in $\dlNames{a}$  as elements in $\Delta^\Inter_A$; 
time points $t\in\NTimepoints$ as $t^\Inter\in\Delta^\Inter_T$;
and intervals $[t,t']\in\NIntervals$ as $[t,t']^\Imc=(t^\Inter,t'^\Inter)\in\Delta^\Inter_{2T}$.
A pair $\tuple{\delta,\epsilon}\in\Delta^\Inter_{A}\times\Delta^\Inter_{\sf a}$ is \emph{well-typed},
if: 
\begin{enumerate}
\addtolength\itemsep{2mm} 
\item $\delta=a^\Inter$ for an attribute `$a$' of value type $\NTimepoints$
and $\epsilon\in\Delta^\Inter_T$; or
\item $\delta=a^\Inter$ for an attribute `$a$' of value type $\NIntervals$ and $\epsilon\in\Delta^\Inter_{2T}$; or
\item $\delta=a^\Inter$ for an attribute `$a$' of value type $\dlNames{a}$ and $\epsilon\in\Delta^\Inter_A$.
\end{enumerate}
Let $\asdom{\Inter}$ be the set of all finite sets of well-typed pairs.
The function $\cdot^\Inter$ maps concept names $A\in\NC$ to $A^\Inter\subseteq\Delta^\Inter_{\sf i}\times\asdom{\Inter}$ and role names
$R\in\NR$ to $R^\Inter\subseteq\Delta^\Inter_{\sf i}\times\Delta^\Inter_{\sf i}\times\asdom{\Inter}$.
The semantics of terms is given by variable assignments, which
 for a time-sorted interpretation $\Inter$ is defined as a function $\Zuweisung$ that maps
\begin{itemize}
\item set variables $X\in \NU$ to finite binary relations $\Zuweisung(X)\in\asdom{\Inter}$, and
\item object variables $x \in \variables{\NI}\cup\variables{\NTimepoints}\cup\variables{\NIntervals}$
to elements $\Zuweisung(x)\in\Delta^\Inter_I\cup\Delta^\Inter_T\cup\Delta^\Inter_{2T}$ (respecting their types).
\end{itemize}
For (set or object) variables $x$, we define $x^{\Inter,\Zuweisung}\defeq\Zuweisung(x)$, and
for abstract individuals, time points, or time intervals $a$, we define $a^{\Inter,\Zuweisung}\defeq a^\Inter$.
The semantics of specifiers is as in Section~\ref{sec:attributedDLs}
with the difference that values can also be time points and intervals: 
\begin{itemize}
\item $\closedAnnotation{a\syntaxequality v}^{\Inter,\Zuweisung}\defeq\{\{\tuple{a^\Inter, v^{\Inter,\Zuweisung}}\}\}$,
with $v\in\valtype(a)\cup\variables{\valtype(a)}$;
\item   
 $\closedAnnotation{a\syntaxequality
  [\valueleft,\valueright]}^{\Inter,\Zuweisung}\defeq\{\{\tuple{a^\Inter, 
 (\valueleft^{\Inter,\Zuweisung}, \valueright^{\Inter,\Zuweisung})} \}\}$, 
with $\valtype(a)=\NIntervals$, 
 and  
 $\valueleft,\valueright\in\NTimepoints\cup\variables{\NTimepoints}$.
\end{itemize}

We are now ready to formally define the semantics of temporal attributes.

\begin{definition}\label{def_localtoglobal}
Consider a temporal domain $\Delta^\Inter_T$ and a domain $\Delta^\Inter_{\ind}$ of individuals
and a domain $\Delta^\Inter_{\ann}$ of annotations,
and let $(\Inter_i)_{i\in\Delta^\Inter_T}$ be a sequence of (non-temporal) interpretations
with domains $\Delta^\Inter_{\ind}$ and $\Delta^\Inter_{\ann}$,
such that, for all $a\in\NI$, we have $a^{\Inter_i}=a^{\Inter_j}$ for all $i,j\in\Delta^\Inter_T$.
We define a \emph{global interpretation} for $(\Inter_i)_{i\in\Delta^\Inter_T}$
as a time-sorted interpretation $\Inter=\tuple{\Delta^\Inter_{\ind},\Delta^\Inter_{\ann},\cdot^\Inter}$ as follows. 
Let $a^\Inter=a^{\Inter_i}$ for all $a\in\NI$.
For any finite set $F\in\asdom{\Inter}$,
let $F_I\defeq F\cap(\Delta^\Inter_{A}\times\Delta^\Inter_{A})$ denote its abstract part
without any temporal attributes.
For any $A\in\NC$, $\delta\in\Delta^\Inter_{\ind}$, and $F\in\asdom{\Inter}$ with 
$F\setminus F_I\neq \emptyset$, we have
$\tuple{\delta,F}\in A^\Inter$ if and only if 
$\tuple{\delta,F_I}\in A^{\Inter_i}$ for some $i\in\Delta^\Inter_T$, 
and the following conditions hold
for all $\tuple{a^\Inter,x}\in F$:
\begin{itemize}
\item if $a=\dt$, then $\tuple{\delta,F_I}\in A^{\Inter_x}$,
\item if $a=\bef$, then $\tuple{\delta,F_I}\in A^{\Inter_j}$ for some $j<x$,
\item if $a=\aft$, then $\tuple{\delta,F_I}\in A^{\Inter_j}$ for some $j>x$,
\item if $a=\unt$, then $\tuple{\delta,F_I}\in A^{\Inter_j}$ for all $j\leq x$,
\item if $a=\since$, then $\tuple{\delta,F_I}\in A^{\Inter_j}$ for all $j\geq x$,
\item if $a=\between$, then $\tuple{\delta,F_I}\in A^{\Inter_j}$ for some $j\in [x]$,
\item if $a=\during$, then $\tuple{\delta,F_I}\in A^{\Inter_j}$ for all $j\in [x]$,
\end{itemize}
where $[x]$ for an element $x\in\Delta^\Inter_{2T}$ denotes the finite
interval represented by the pair of numbers $x$, and $j\in\Delta^\Inter_T$.
For roles $R\in\NR$, we define $\tuple{\delta,\epsilon,F}\in R^\Inter$ analogously.
\end{definition}

\begin{definition}[Temporal Geometric Interpretation]
\sloppy{A temporal $m$-dimensional $\linearmap$-geometric interpretation 
with temporal domain $\Delta_T$
is a sequence 
$(\eta_j)_{j\in \Delta_T}$ 
of $m$-dimensional $\linearmap$-geometric interpretations. }
An $m$-dimensional $\linearmap$-geometric interpretation $\eta$ 
is \emph{global} for $(\eta_j)_{j\in \Delta_T}$ and $D_\ann\subseteq \dlNames{a}$ 
if $\Imc(\eta, (D_\ann\cup\NTimepoints\cup\NIntervals)^\star)$ is global for $(\Imc(\eta_j,D^\star_\ann))_{j\in \Delta_T}$.
\end{definition}

Let $\NIkb{\Omc}$ denote the union of all elements in $\dlNames{a}$, \NTimepoints, and \NIntervals occurring in \Omc.

\begin{definition}[Geometric Model]
Let \Omc be a \dlliteat ontology. 
An $\linearmap$-geometric interpretation $\eta$ 
is   an $m$-dimensional $\linearmap$-geometric model of \Omc if it is global for
a sequence  
$(\eta_j)_{j\in \Delta_T}$ 
of $m$-dimensional $\linearmap$-geometric interpretations and $\NIkb{\Omc}$, plus 
$\Imc(\eta, \NIkbstar{\Omc})$ satisfies~\Omc. 
\end{definition}

Example~\ref{ex:counterexample} shows that even if temporal specifiers are only of the form $\dt$ and $\during$, convex geometric models may not exist for  satisfiable \dlliteat ontologies.
\begin{example}\label{ex:counterexample}
Let 
\begin{align*}\Omc=&\{\exists R@\closedAnnotation{\during\syntaxequality [1,2]}\sqsubseteq A, 
\exists R@\closedAnnotation{\dt\syntaxequality 1} \sqcap A\sqsubseteq \bot, \\&R(a,a)@\closedAnnotation{\dt\syntaxequality 1},
 R(b,b)@\closedAnnotation{\dt\syntaxequality 1}, R(a,b)@\closedAnnotation{\dt\syntaxequality 2}, R(b,a)@\closedAnnotation{\dt\syntaxequality 2}\}
\end{align*}
 and let $\eta$ be a convex $f$-geometric model of $\Omc$. 
Let $\delta=0.5\eta(a)+0.5\eta(b)$. By the convexity of $\eta(R@\closedAnnotation{\dt\syntaxequality 1})$ and $\eta(R@\closedAnnotation{\dt\syntaxequality 2})$, 
we have that 
\begin{align*}
f(\delta,\delta)\in \eta(R@\closedAnnotation{\dt\syntaxequality 1})\text{ and }
f(\delta,\delta)\in \eta(R@\closedAnnotation{\dt\syntaxequality 2}),
\end{align*}
 so $\Imc(\eta, \NIkbstar{\Omc})\models R(\delta,\delta)@\closedAnnotation{\dt\syntaxequality 1}$ and $\Imc(\eta, \NIkbstar{\Omc})\models R(\delta,\delta)@\closedAnnotation{\dt\syntaxequality 2}$. 
It follows that $\Imc(\eta, \NIkbstar{\Omc})\models R(\delta,\delta)@\closedAnnotation{\during\syntaxequality [1,2]}$. 
Since $\Imc(\eta, \NIkbstar{\Omc})\models R@\closedAnnotation{\during\syntaxequality [1,2]}\sqsubseteq A$, we also have that $\Imc(\eta, \NIkbstar{\Omc})\models A(\delta)$. 
Hence $\Imc(\eta, \NIkbstar{\Omc})\not\models \exists R@\closedAnnotation{\dt\syntaxequality 1}\sqcap A\sqsubseteq \bot$. This means that
 $\Imc(\eta, \NIkbstar{\Omc})$ is not a model of $\Omc$. \hfill {\mbox{$\triangleleft$}} 
\end{example}

To overcome this problem, we introduce a restriction on the specifiers allowed on roles. 
We introduce \emph{atemporal specifiers}. An atemporal specifier is a specifier $S$ that can only be interpreted as a set $S^{\Inter,\Zuweisung}\subseteq\asdom{\Inter}$ of matching annotation sets that do not contain any temporal attribute.

To show that convex geometric models can capture some \dlliteat ontologies, we   use   concept inclusions with conjunctions on the left-hand side, which can be expressed in \dlliterhorn. 
The following example shows that adding conjunction in role inclusions may however lead to satisfiable ontologies not having a convex model, even for plain DLs. 
\begin{example}
Assume role conjunctions are allowed in the ontology. 
Let 
\begin{equation*}
\begin{array}{l}
\Omc=\{R_1\sqcap R_2\sqsubseteq R_3, \exists R_1\sqsubseteq A, \exists R_3 \sqcap A\sqsubseteq \bot, R_1(a,a), R_1(b,b), R_2(a,b), R_2(b,a)\}
 \end{array}
\end{equation*}  
and let $\eta$ be a convex $f$-geometric model of $\Omc$. 
Let $\delta=0.5\eta(a)+0.5\eta(b)$. By the convexity of $\eta(R_1)$ and $\eta(R_2)$, we have that 
\begin{align*}
f(\delta,\delta)\in \eta(R_1)\text{ and } 
f(\delta,\delta)\in \eta(R_2),\text{ hence }f(\delta,\delta)\in \eta(R_1\sqcap R_2).
\end{align*}
Then since $\eta$ is a model of $R_1\sqcap R_2\sqsubseteq R_3$, $f(\delta,\delta)\in \eta(R_3)$ so $\delta\in \eta(\exists R_3)$. Moreover, since $\eta$ is a model of $\exists R_1\sqsubseteq A$, and $\delta\in \eta(\exists R_1)$, $\delta\in \eta(A)$. 
Hence $\eta\not\models \exists R_3\sqcap A\sqsubseteq\bot$ so $\eta$ is not a model of $\Omc$. 
\hfill {\mbox{$\triangleleft$}}
\end{example}

We now state the main result of this section, which 
states that, under certain conditions, satisfiable
\dlliteat ontologies
have a convex geometric model.
The need of Condition (i) is already illustrated by Example~\ref{ex:counterexample}
whereas Condition (ii) ensures that
the underlying logic is Horn (that is, it does not have disjunctions which 
can be expressed with the temporal attributes \bef, \aft and \between). 

\begin{restatable}{theorem}{temporalgeometric}\label{prop:addingtime}
Let \Omc be a satisfiable \dlliteat ontology without negative role inclusions and such that (i) all specifiers attached to a role 
in \Omc are atemporal, and (ii) \bef, \aft and \between do not occur in \Omc.  
Then \Omc has a convex geometric model. 
\end{restatable}

\section{Related Work}\label{sec:relatedwork}
 
 Traditionally, most KG embedding models are time-unaware. 
 These models embed both entities and relations in a low-dimensional 
 latent space based on some regularities of target KG. 
 They can be used as   approximate reasoning methods~\cite{PaTh07,PRZ2016} 
 for completing KG without using the schema. Typical 
 embedding models include the translation based models, such as TransE~\cite{bordes2013translating} and 
 bilinear models, such as 
 ComplEx~\cite{TrouillonDGWRB17}  and SimplE~\cite{Kazemi018}. 
  From the expressiveness perspective, TransE and DisMult have been shown to be not fully expressive; however, CompleEx and 
  SimplE are fully expressive. Gutierrez-Basulto and Schockaert~\cite{DBLP:conf/kr/Gutierrez-Basulto18}
 use geometric models to study the compatibility between TBox/ontology and KG embeddings. They show that bilinear models 
 (inc. ComplEx and SimplE) cannot strictly represent relation subsumption rules. Wiharja et al.~\cite{WPKD2020} show 
 that many well known KG embeddings based on KG completion methods are not impressive, when schema aware correctness is 
 considered, despite good performance reported in silver standard based evaluations.  
Currently, more and more applications are involving dynamic KG, where 
knowledge in practice is time-variant and consists of sequences of observations. For example, in  
  recommendation systems based on KG, new items and new user actions appear in real time. 
  Accordingly, temporal KG embedding models incorporate time information into their node and relation representations. 
We next discuss how temporal 
information is taken into account in KG embeddings and how 
it has been used in combination with classical DLs.

\mypar{Temporal Knowledge Graph Embeddings.} 
Temporal KG embedding models can be seen as extensions of static KG embedding models. 
A basic approach is to collapse the dynamic graph into a static graph by aggregating the temporal observations over time~\cite{LiKl2007}. 
 This  approach, however,  may lose large amounts of information. An alternative approach is to  
  give more weights to snapshots that are more recent~\cite{Sharan08temporal-relationalclassifiers}. 
Another alternative approach to aggregating temporal observations is to apply decomposition methods to dynamic graphs. 
The idea is to model a KG as an order 4 tensor and decompose it using CP or Tucker, or other decomposition methods to obtain entity, relation, and timestamp embeddings~\cite{ETYBK2016}.
In addition to aggregation based approaches, there are approaches extending static KG embedding, such as TransE, 
   by adding a timestamp embedding into the score function~\cite{jiang-etal-2016-towards,MTD2018}.  
   Jiang et  al.~\cite{jiang-etal-2016-towards} only use such timestamps to maintain temporal order, 
   while using Integer Linear Programming  to encode the temporal consistency information as constraints.  
   Ma et al.~\cite{MTD2018} extend several models (Tucker, RESCAL, HolE, ComplEx, DistMult) by adding a timestamp embedding to their score functions. These models may not work well when the number of timestamps is large. Furthermore, since they only learn embeddings for observed timestamps, they cannot generalize to unseen timestamps. 
Dasgupta et al.~\cite{dasgupta-etal-2018-hyte} fragments  a  temporally-scoped  input  KG 
 into multiple static subgraphs with each subgraph corresponding to a timestamp.   
There are also approaches of applying random walk models for temporal KG. 
E.g., Bian et al~\cite{BKDD2019} use metapath2vec to generate random walks on both the initial KG and the updated nodes and re-compute the embeddings for these nodes.  These approaches mainly leverage the temporal aspect of dynamic graphs to reduce the computations. However, they may fail at capturing the evolution and the temporal patterns of the nodes. 
Another  natural choice for modeling temporal KG is by extending sequence models to graph data. E.g., García-Duran~\cite{GDN2018} extend TransE and DistMult by combining the relation and timestamp through a character LSTM, so as to learn representations for time-augmented KG facts that can be used  in  conjunction  with  existing  scoring  functions for link prediction.
Ma et al.~\cite{MTD2018}  argue that temporal KG embeddings could also  be used as models   for cognitive episodic memory (facts we remember and can recollect) and for  semantic memory (current facts we know) can be generated from episodic memory by a marginalization operation.

\mypar{Temporal Description Logics.} 
 In the DL literature, there are several approaches for representing and reasoning  temporal 
 information~\cite{DBLP:conf/time/LutzWZ08,wolter1999temporalizing,DBLP:journals/tocl/ArtaleKRZ14}. 
  Schmiedel~\cite{Schmiedel1990}  was the first to propose an extension of the  description logics (the $\mathcal{FLENR^-}$ DL in this case) with an interval-based temporal logic, with the temporal quantifier  {\bf at}, the existential and universal temporal quantifiers {\bf sometime} and {\bf alltime}.  
 Artale and Franconi~\cite{ArFr1994,ArFr1998}  considered a class of interval-based temporal description logics  by reducing  the expressivity to keep the property of decidability of the logic proposed by Schmiedel~\cite{Schmiedel1990}. 
 Schild  proposed $\mathcal{ALCT}$~\cite{Schi1993}, extending  $\mathcal{ALC}$  with point-based modal temporal connectives from tense logic~\cite{Burgess1984}, including existential future ($\lozenge$), universal future ($\square$), next instant ($\bigcirc$),until ($\mathcal{U}$), reflexive until ({\bf U}). Wolter and Zakharyaschev   studied the $\mathcal{ALC_M^-}$ DL and showed that it is decidable in the class of linear, discrete and unbounded temporal structures~\cite{WoZa1998}. They also showed that the $\mathcal{ALC_M}$ DL (extending the $\mathcal{ALC_M^-}$ DL with global roles) is undecidable~\cite{WoZa1999}. 
 Temporal operators can be used in a temporal ABox as well, allowing the use of next instant ($\bigcirc\varphi$) 
 and until ($\varphi\mathcal{U}\psi$) with ABox assertions~\cite{AKLWZ2007}. 
Ozaki et al.~\cite{DBLP:conf/dlog/OzakiKR18,DBLP:conf/birthday/OzakiKR19} propose temporally attributed DLs, 
 which allows the use of absolute temporal information in both TBoxes and ABoxes. 
 They show that the satisfiability of ground $\mathcal{ELH^\mathbb{T}_@}$ ontologies is ExpTime-complete, and 
 that the satisfiability of ground $\mathcal{ELH^\mathbb{T}_@}$ ontologies without the 
 temporal attributes \textsf{between}, \textsf{before} and \textsf{after} is PTime-complete. 

 \section{Conclusion}\label{sec:conclusion}
 We investigate how geometric models can (or cannot) be used to capture rules about annotated data expressed in the formalism of attributed DLs. We show that every satisfiable attributed \dlliterhorn ontology has a convex geometric model and that this is also the case when allowing the use of temporal attributes under some restrictions.  
There is still a long way to make this result practical 
since we still require an 
embedding technique that would construct such a model. 
In this direction, we highlight the work of Abboud et al.~\cite{DBLP:conf/nips/AbboudCLS20}, where relations are mapped
to convex regions in the format of hyper-rectangles.
 
\mypar{Acknowledgements.} 
We thank Bruno Figueira Lourenço for his contribution on the proof of Theorem~\ref{thm:linearmapTh}. 
Ozaki is supported by the Norwegian Research Council, grant number 316022.

\bibliographystyle{plain}
\bibliography{referencesshort}

\appendix

\section{Proof of Theorem 3}

Our definitions of geometric interpretations and models
 are based on the work by Gutiérrez-Basulto and Schockaert~\cite{DBLP:conf/kr/Gutierrez-Basulto18}.
The main difference is that in 
Definition~\ref{def:sat-assertions}
we allow any kind of linear map that respects some conditions (in particular, (i) to (iii)), instead of using vector concatenation for  representing tuples in the vector space. 
Here we show 
that our  conditions for the linear maps
are harmless in the sense that all results obtained in the mentioned 
work still hold with the more general conditions. 

Let $f: \mathbb{R}^m\times\mathbb{R}^n\mapsto \mathbb{R}^{m+n}$ be a linear map satisfying the following: 
\begin{enumerate}[label=\textit{(\roman*)}]
\item \label{item1} the restriction of $f$ to $\mathbb{R}^m\times \{0\}^{n}$ is injective;
\item \label{item2} the restriction of $f$ to $ \{0\}^{m}\times \mathbb{R}^n$ is injective;
\item  \label{item3} $f(\mathbb{R}^m\times \{0\}^n)\cap f(\{0\}^m\times \mathbb{R}^n)=\{0^{m+n}\}$;
\end{enumerate} 
where $0^{n}$ denotes the vector $(0,\ldots,0)$ with $n$ zeros.

\begin{proposition}\label{prop:bijective}
Any linear map satisfying \ref{item1}, \ref{item2}, and \ref{item3} is bijective.
\end{proposition}
\begin{proof}
Let $f$ be a linear map satisfying \ref{item1}, \ref{item2}, and \ref{item3}. 
\begin{claim}
The \emph{kernel of $f$} $\{x\in\mathbb{R}^m\times\mathbb{R}^n\mid f(x)=0^{m+n}\}$ is
equal to $\{0^{m+n}\}$.
\end{claim}

\noindent
\textit{Proof of the Claim.}
Suppose that $(x,y)$
belongs to the kernel of $f$, where  $x\in \mathbb{R}^m$ and $y\in\mathbb{R}^n$. Then,
$f(x,y)=0^{m+n}$.
Since $f$ is a linear map, 
addition is preserved, so $f(x,0^n)+f(0^m,y)=f(x,y)=0^{m+n}$. Because of~\ref{item3}, we conclude that  
$f(x,0^{n})=f(0^{m},y)=0^{m+n}$. By \ref{item1}-\ref{item2}, $x=0^m$ and $y=0^n$, as required. 

\medskip

Suppose $f(x)=f(y)$. In other words, $0^{m+n}=f(x)-f(y)$. 
Since $f$ is a linear map, $0^{m+n}=f(x)-f(y)=f(x-y)$. This means that
$x-y$ is in the kernel of $f$ and, 
by the Claim, $x-y=0^{m+n}$. Thus, $x=y$, which means that $f$ is injective.

By the Claim, the dimension of the kernel of $f$ is $0$. Since the dimension
of the domain of $f$ is $m+n$, by the rank-nullity theorem, the dimension
of the image of $f$ must be  $m+n$. This means that   $f$ is also surjective.
\qed
\end{proof}

By Proposition~\ref{prop:bijective},  any linear map satisfying
\ref{item1}, \ref{item2}, and \ref{item3} is bijective and, by the next proposition, 
any such linear map could be used in the formalisation as it would yield  isomorphic  regions.  
Given $S\subseteq \mathbb{R}^n$, we denote by ${\sf co}(S)$ the convex hull of $S$.

\begin{proposition}\label{prop:isomorphism}
Let $f: \mathbb{R}^m\times\mathbb{R}^n\mapsto \mathbb{R}^{m+n}$ and
$g: \mathbb{R}^m\times\mathbb{R}^n\mapsto \mathbb{R}^{m+n}$ be bijective linear maps.
Let $h:= f\circ g^{-1}$. Then, for every $S\subseteq \mathbb{R}^m\times\mathbb{R}^n$,
we have $$h({\sf co}(f(S)))={\sf co}(g(S)).$$
\end{proposition}

\begin{proof}
Let $S\subseteq \mathbb{R}^m\times\mathbb{R}^n$ and let $x \in {\sf co}(f(S))$.
Then, there are $\alpha_1,\dots,\alpha_k\in [0,1]$ and $s_1,\ldots,s_k\in S$ such that
$$x= \sum^k_{i=1}\alpha_i f(s_i), \qquad  \sum^k_{i=1} \alpha_i =1.$$
By definition of $h$ and its linearity, we have
$$h(x)=\sum^k_{i=1}\alpha_i g(s_i)$$
therefore $h(x)\in {\sf co} (g(S))$.
Conversely, if $x \in {\sf co}( g (S))$, an analogous argument shows that $h^{-1}(x) \in f(S)$.
This shows that,
indeed, $h({\sf co}(f(S)))={\sf co}(g(S))$. \qed
\end{proof}

\section{Proof of Theorem 4}

\lemfundam*
\begin{proof}
We start with the case where $\alpha$ is a concept assertion. 
\begin{enumerate}
\item $\Imc(\eta, D^\star)\models \syntaxpatom{A(a)}{S}$ iff $(a^{\Imc(\eta, D^\star)},  F) \in A^{\Imc(\eta, D^\star)}$ for some $F\in S^{\Imc(\eta, D^\star)}$. 
\item $a^{\Imc(\eta, D^\star)}=\eta(a)$.
\item Since $S$ is closed, $S^{\Imc(\eta, D^\star)}=\{F^\star_S\}$. 
\item It follows from (1), (2) and (3) that $\Imc(\eta, D^\star)\models \syntaxpatom{A(a)}{S}$ iff $(\eta(a)  ,F^\star_S) \in A^{\Imc(\eta, D^\star)}$. 
\item By definition of $\Imc(\eta, D^\star)$, $(\eta(a)  ,F^\star_S) \in A^{\Imc(\eta, D^\star)}$ iff $\eta(a) \in \eta(\syntaxpatom{A}{S})$, \ie $\eta\models \syntaxpatom{A(a)}{S}$.
\end{enumerate}
The proof for the case where $\alpha$ is a role assertion is similar.
\\

We next show the case where $\alpha$ is a concept inclusion, \ie  
$\alpha$ is either of the form $\bigsqcap_{i=1}^n \syntaxpatom{B_i}{S_i}\sqsubseteq \syntaxpatom{B}{S}$ or $\bigsqcap_{i=1}^n \syntaxpatom{B_i}{S_i}\sqsubseteq \bot$ where $B_i$ and $B$ are of the form $A_i$ or $\exists P_i$.\smallskip

Assume that $\Imc(\eta, D^\star)\models \alpha$. 
Let $\delta\in \mathbb{R}^m$ be such that $\delta\in \eta(\bigsqcap_{i=1}^n\syntaxpatom{B_i}{S_i})=\bigcap_{i=1}^n\eta(\syntaxpatom{B_i}{S_i})$. 
For $1\leq i\leq n$: 
\begin{enumerate}
    \item $\delta\in \eta(\syntaxpatom{B_i}{S_i})$ so $(\delta ,F^\star_{S_i})\in B_i^{\Imc(\eta,D^\star)}$;
    \item $F^\star_{S_i}\in S_i^{{\Imc(\eta,D^\star)}, \Zmc}$ for every variable assignment $\Zmc$;
    \item (1) and (2) implies that $\delta\in (\syntaxpatom{B_i}{S_i})^{{\Imc(\eta,D^\star)}, \Zmc}$ for every variable assignment $\Zmc$. 
\end{enumerate}
It follows that $\delta\in (\bigsqcap_{i=1}^n\syntaxpatom{B_i}{S_i})^{{\Imc(\eta,D^\star)}, \Zmc}$. 
\begin{itemize}
    \item Assume that $\alpha$ is of the form $\bigsqcap_{i=1}^n \syntaxpatom{B_i}{S_i}\sqsubseteq \syntaxpatom{B}{S}$. 
    \begin{enumerate}
        \item Since $\Imc(\eta, D^\star)\models\alpha$, then $\delta\in (\syntaxpatom{B}{S})^{{\Imc(\eta,D^\star)}, \Zmc}$. 
        \item Hence there exists $F\in S^{\Inter,\Zuweisung}$ such that $(\delta, F)\in B^{{\Imc(\eta,D^\star)}}$.
        \item By definition of $B^{{\Imc(\eta,D^\star)}}$, it follows that $\delta\in \eta(\syntaxpatom{B}{S_F})$ with $S_F$ a ground (closed or open) specifier that contains exactly the pairs $a\syntaxequality b$ such that $(a,b)$ are in $F$, except $(\star,\star)$.
        \item  By (2), $F\in S^{\Inter,\Zuweisung}$ so the attribute-value pairs in $S_F$ form a superset of those in $S$ and either $S_F=S$ if $S$ is closed, or $S_F\implies S$ if $S$ is open. In both cases, $\eta(\syntaxpatom{B}{S_F})\subseteq  \eta(\syntaxpatom{B}{S})$. 
    \item By (3) and (4), $\delta\in \eta(\syntaxpatom{B}{S})$.
    \end{enumerate}
    It follows that $\eta(\bigsqcap_{i=1}^n\syntaxpatom{B_i}{S_i})\subseteq \eta(\syntaxpatom{B}{S})$
\item Assume that $\alpha$ is of the form $\bigsqcap_{i=1}^n \syntaxpatom{B_i}{S_i}\sqsubseteq \bot$.  Since $\Imc(\eta, D^\star)\models\alpha$, then $\delta\in\emptyset$. This is a contradiction so $\eta(\bigsqcap_{i=1}^n\syntaxpatom{B_i}{S_i})$ must be empty. Hence 
$$\eta(\bigsqcap_{i=1}^n\syntaxpatom{B_i}{S_i})\subseteq \eta(\bot).$$
\end{itemize}
We conclude that $\eta\models \alpha$.\smallskip

In the other direction, assume that $\eta\models \alpha$. Let $\delta\in (\bigsqcap_{i=1}^n\syntaxpatom{B_i}{S_i})^{\Imc(\eta,D^\star)}=\bigcap_{i=1}^n\syntaxpatom{B_i}{S_i}^{\Imc(\eta,D^\star)}$. 
For $1\leq i\leq n$:
\begin{enumerate}
    \item There exists $(\delta, F_i)\in B_i^{\Imc(\eta,D^\star)}$ such that $F_i\in S_i^{\Imc(\eta,D^\star),\Zmc}$ (for any $\Zmc$).
    \item By (1) and the definition of $\Imc(\eta,D^\star)$, $\delta\in \eta(\syntaxpatom{B_i}{T_i})$ for some $T_i$ that contains exactly the attribute value pairs in $F_i$, except $(\star,\star)$, and such that $T_i$ is closed if $(\star,\star)\notin F_i$, and open otherwise.
    \item By (1), $F_i\in S_i^{\Imc(\eta,D^\star),\Zmc}$ so by definition of $T_i$ in (2), it is easy to check that when $S_i$ is closed $T_i=S_i$ and when $S_i$ is open so is $T_i$ and $T_i\implies S_i$. Hence, in both cases $\eta(\syntaxpatom{B_i}{T_i})\subseteq \eta(\syntaxpatom{B_i}{S_i})$
\end{enumerate}
By (2) and (3) $\delta\in \eta(\bigsqcap_{i=1}^n\syntaxpatom{B_i}{S_i})$. 
\begin{itemize}
    \item Assume that $\alpha$ is of the form $\bigsqcap_{i=1}^n \syntaxpatom{B_i}{S_i}\sqsubseteq \syntaxpatom{B}{S}$. Since $\eta\models \alpha$, it follows that $\delta\in \eta(\syntaxpatom{B}{S})$. By construction of $\Imc(\eta,D^\star)$, it follows that $(\delta, F_S^\star)\in B^{\Imc(\eta,D^\star)}$. Since $F_S^\star\in S^{\Imc(\eta,D^\star),\Zmc}$ for any $\Zmc$, $\delta\in (\syntaxpatom{B}{S})^{\Imc(\eta,D^\star)}$. 
    Hence 
    $$(\bigsqcap_{i=1}^n\syntaxpatom{B_i}{S_i})^{\Imc(\eta,D^\star)}\subseteq (\syntaxpatom{B}{S})^{\Imc(\eta,D^\star)}.$$
    \item Assume that $\alpha$ is of the form $\bigsqcap_{i=1}^n \syntaxpatom{B_i}{S_i}\sqsubseteq \bot$. Since 
    $\eta\models \alpha$, then $\delta\in\emptyset$. This is a contradiction so 
    $(\bigsqcap_{i=1}^n\syntaxpatom{B_i}{S_i})^{\Imc(\eta,D^\star)}$ must be empty. 
    Hence $$(\bigsqcap_{i=1}^n\syntaxpatom{B_i}{S_i})^{\Imc(\eta,D^\star)}\subseteq \bot^{\Imc(\eta,D^\star)}.$$
\end{itemize}
We conclude that $\Imc(\eta,D^\star)\models \alpha$.\\

Finally, we show the case where $\alpha$ is a role inclusion, \ie $\alpha$ is either of the form $\syntaxpatom{P}{S}\sqsubseteq \syntaxpatom{Q}{T}$ or $\syntaxpatom{P}{S}\sqsubseteq \neg\syntaxpatom{Q}{T}$, where $P$ and $Q$ are of the form $R$ or $R^-$ for some $R\in \NR$.\smallskip

Assume that $\Imc(\eta, D^\star)\models \alpha$. 
Let $\delta_1,\delta_2\in \mathbb{R}^m$ be such that $f(\delta_1,\delta_2)\in \eta(\syntaxpatom{P}{S})$. 
By construction of $\Imc(\eta,D^\star)$, $(\delta_1, \delta_2 ,F^\star_{S})\in P^{\Imc(\eta,D^\star)}$. 
Moreover, $F^\star_{S}\in S^{{\Imc(\eta,D^\star)}, \Zmc}$ for every variable assignment $\Zmc$. Hence, $(\delta_1, \delta_2)\in (\syntaxpatom{P}{S})^{{\Imc(\eta,D^\star)}, \Zmc}$ for every $\Zmc$. 
\begin{itemize}
    \item Assume that $\alpha$ is of the form $\syntaxpatom{P}{S}\sqsubseteq \syntaxpatom{Q}{T}$. 
    \begin{enumerate}
        \item Since $\Imc(\eta, D^\star)\models\alpha$, $(\delta_1, \delta_2)\in (\syntaxpatom{Q}{T})^{{\Imc(\eta,D^\star)}, \Zmc}$.
        \item By (1), there exists $F\in T^{\Inter,\Zuweisung}$ such that $(\delta_1, \delta_2, F)\in Q^{{\Imc(\eta,D^\star)}}$.
        \item By (2) and definition of $Q^{{\Imc(\eta,D^\star)}}$, it follows that $f(\delta_1, \delta_2)\in \eta(\syntaxpatom{Q}{S_F})$ with $S_F$ a ground (closed or open) specifier that contains exactly the pairs $a\syntaxequality b$ such that $(a,b)$ are in $F$, except $(\star,\star)$. 
        \item By (2) $F\in T^{\Inter,\Zuweisung}$ so the attribute-value pairs in $S_F$ defined in (3) form a superset of those in $T$ and either $S_F=T$ if $T$ is closed, or $S_F\implies T$ if $T$ is open. 
In both cases, $\eta(\syntaxpatom{Q}{S_F})\subseteq  \eta(\syntaxpatom{Q}{T})$. 
    \item By (3) and (4), $f(\delta_1, \delta_2)\in \eta(\syntaxpatom{Q}{T})$. 
    \end{enumerate}   
It follows that $\eta(\syntaxpatom{P}{S})\subseteq \eta(\syntaxpatom{Q}{T})$.
    \item Assume that $\alpha$ is of the form $\syntaxpatom{P}{S}\sqsubseteq \neg\syntaxpatom{Q}{T}$. 
    \begin{enumerate}
        \item Since $\Imc(\eta, D^\star)\models\alpha$, $(\delta_1, \delta_2)\in (\mathbb{R}^{m} \times \mathbb{R}^{m})\setminus(\syntaxpatom{Q}{T})^{{\Imc(\eta,D^\star)}, \Zmc}$.
        \item By (1), for every $F\in T^{\Inter,\Zuweisung}$, $(\delta_1, \delta_2, F)\notin Q^{{\Imc(\eta,D^\star)}}$. 
        \item (2) implies that $f(\delta_1, \delta_2)\notin \eta(\syntaxpatom{Q}{T})$ (otherwise we would have $(\delta_1, \delta_2 ,F^\star_T)\in Q^{{\Imc(\eta,D^\star)}}$ and $F^\star_T\in T^{\Inter,\Zuweisung}$). Hence $f(\delta_1, \delta_2)\in \mathbb{R}^{2\cdot m} \setminus \eta(\syntaxpatom{Q}{T}) =\eta(\neg\syntaxpatom{Q}{T})$.
    \end{enumerate} 
  It follows that $\eta(\syntaxpatom{P}{S})\subseteq \eta(\neg\syntaxpatom{Q}{T})$.
\end{itemize}
We conclude that $\eta\models \alpha$.
\smallskip

In the other direction, assume that $\eta\models \alpha$. Let $(\delta_1,\delta_2)\in (\syntaxpatom{P}{S})^{\Imc(\eta,D^\star)}$. There exists $(\delta_1,\delta_2, F)\in P^{\Imc(\eta,D^\star)}$ such that $F\in S^{\Imc(\eta,D^\star),\Zmc}$ (for any $\Zmc$). By construction of $\Imc(\eta,D^\star)$, it follows that $f(\delta_1,\delta_2)\in \eta(\syntaxpatom{P}{S'})$ for some $S'$ that contains exactly the attribute value pairs in $F$, except $(\star,\star)$, and such that $S'$ is closed if $(\star,\star)\notin F$, and open otherwise. It is easy to check that when $S$ is closed $S'=S$ and when $S$ is open so is $S'$ and $S'\implies S$. Hence, in both cases $\eta(\syntaxpatom{P}{S'})\subseteq \eta(\syntaxpatom{P}{S})$, so that $f(\delta_1,\delta_2)\in \eta(\syntaxpatom{P}{S})$. 

\begin{itemize}
\item Assume that $\alpha$ is of the form $\syntaxpatom{P}{S}\sqsubseteq \syntaxpatom{Q}{T}$. 
\begin{enumerate}
    \item Since $\eta\models\alpha$, $f(\delta_1,\delta_2)\in\eta(\syntaxpatom{Q}{T})$.
    \item By (1) and construction of $\Imc(\eta,D^\star)$, $(\delta_1,\delta_2, F_T^\star)\in Q^{\Imc(\eta,D^\star)}$.
    \item By (2) and since $F_T^\star\in T^{\Imc(\eta,D^\star),\Zmc}$ for any $\Zmc$, $(\delta_1,\delta_2)\in (\syntaxpatom{Q}{T})^{\Imc(\eta,D^\star)}$.
\end{enumerate} 
 It follows that $(\syntaxpatom{P}{S})^{\Imc(\eta,D^\star)}\subseteq (\syntaxpatom{Q}{T})^{\Imc(\eta,D^\star)}$.

\item Assume that $\alpha$ is of the form $\syntaxpatom{P}{S}\sqsubseteq \neg\syntaxpatom{Q}{T}$. 
\begin{enumerate}
    \item Since $\eta\models\alpha$, $f(\delta_1,\delta_2)\in\mathbb{R}^{2\cdot m} \setminus\eta(\syntaxpatom{Q}{T})$.
    \item (1) implies $(\delta_1,\delta_2)\notin (\syntaxpatom{Q}{T})^{\Imc(\eta,D^\star)}$. Indeed, otherwise there would exist $F\in T^{\Imc(\eta,D^\star),\Zmc}$ such that $(\delta_1,\delta_2, F)\in (Q)^{\Imc(\eta,D^\star)}$: in the case where $T$ is closed $F$ would be equal to $F_T^\star$ and it would imply that $f(\delta_1,\delta_2)\in\eta(\syntaxpatom{Q}{T})$. In the case where $T$ is open, $F$ would be such that there exists $S_F\implies T$ such that  $f(\delta_1,\delta_2)\in\eta(\syntaxpatom{Q}{S_F})$ and thus $f(\delta_1,\delta_2)\in\eta(\syntaxpatom{Q}{S_F})\subseteq \eta(\syntaxpatom{Q}{T})$.
    \item By (2), $(\delta_1,\delta_2)\in (\mathbb{R}^{m} \times \mathbb{R}^{m})\setminus (\syntaxpatom{Q}{T})^{\Imc(\eta,D^\star)}=(\neg\syntaxpatom{Q}{T})^{\Imc(\eta,D^\star)}$. 
\end{enumerate}
It follows that $(\syntaxpatom{P}{S})^{\Imc(\eta,D^\star)}\subseteq (\neg\syntaxpatom{Q}{T})^{\Imc(\eta,D^\star)}$
\end{itemize}
We conclude that $\Imc(\eta,D^\star)\models \alpha$.
\end{proof}

\section{Proof of Theorem 6}

\dlliteaconvexmodel*

\begin{proof}
Let \Omc be a satisfiable \dllitea ontology without negative role inclusions. We use grounding to translate \Omc into an equisatisfiable \dlliterhorn ontology. Let  \Jmc be an interpretation such that $\Delta^\Jmc_\ind=\dlNames{i}$,  $\Delta^\Jmc_\ann=\dlNames{a}$ and $a^\Jmc=a$ for all $a\in\NI$. 
Let $\Zuweisung$ be a variable assignment 
mapping object variables $x \in \objectvariables$ to elements
$\Zuweisung(x) \in \NIkb{\Omc}$ 
and set variables  $X\in \setvariables$ to finite binary 
relations $\Zuweisung(X) \in \asdom{\NIkb{\Omc}}$.  
Consider a concept or 
role inclusion $I$ of the form 
  $\syntaxsetrestriction{X_1}{S_1}, \ldots, \syntaxsetrestriction{X_n}{S_n} \quad  (\bigsqcap_{i=1}^k K_i\sqsubseteq L)$.
  A variable assignment $\Zuweisung$ is said to be compatible 
with $I$ if $\Zuweisung(X_1)\in S_1^{\Jmc,\Zuweisung}, \dots, \Zuweisung(X_n)\in S_n^{\Jmc,\Zuweisung}$.
   The \emph{$\Zuweisung$-instance} $I_\Zuweisung$ of $I$ is the 
concept or role inclusion $\bigsqcap_{i=1}^k K_i'\sqsubseteq L'$ obtained by:
\begin{itemize} 
\item replacing each   $X$ with $\closedAnnotation{a \syntaxequality b\mid  \tuple{a,b}\in\Zuweisung(X) }$; 
\item replacing each assignment $a \syntaxequality \syntaxintent{X}{b}$ 
occurring in some specifier by all assignments $a \syntaxequality c$ such that $\tuple{b,c}\in\Zuweisung(X)$; and, 
\item replacing each object variable $ x$ by $\Zuweisung(x)$.
\end{itemize}

Let $\Omc_{\sf g}$ contains 
\begin{itemize}
\item 
all assertions $\syntaxpatom{E(\vec{t})}{S}$ of $\Omc$, and 
\item all $\Zuweisung$-instances $I_\Zuweisung$ for all 
concept or role inclusions $I$ in $\Omc$ and all compatible variable assignments $\Zuweisung$ (they are a finite number since $\NIkb{\Omc}$ is finite). 
\end{itemize}
The ground ontology $\Omc_{\sf g}$ can be translated into a 
standard \dlliterhorn ontology $\dltrans{\Omc_{\sf g}}$ as follows: 
replace every annotated concept/role name 
$\syntaxpatom{E}{S}$ (or inverse role $\syntaxpatom{R^-}{S}$) with a fresh concept/role name $E_S$ (or inverse role $R_S^-$) 
in all the assertions and concept or role inclusions 
of 
$\Omc_{\sf g}$, 
and extend the obtained \dlliterhorn ontology 
$\dltrans{\Omc_{\sf g}}$ by all axioms $E_S\sqsubseteq E_T$ where $S$ and $T$ are ground specifiers built over $\NIkb{\Omc}$ and $S\implies T$.

Since $\dltrans{\Omc_{\sf g}}$ is a satisfiable \dlliterhorn ontology, it has a finite model $\Jmc'$, which can be chosen such that $\Delta^{\Jmc'}=\dlNames{i}$, and $a^{\Jmc'}=a$ for every $a\in\dlNames{i}$. Moreover, since  $\dltrans{\Omc_{\sf g}}$ does not contain negative role inclusion, it is a quasi-chained ontology. By Proposition 3 in \cite{DBLP:conf/kr/Gutierrez-Basulto18}, it follows that $\dltrans{\Omc_{\sf g}}$ has a convex $m$-dimensional $\linearmap$-geometric model $\eta'$ such that $\{E(\vec{t})\mid \eta'\models E(\vec{t})\}=\{E(\vec{t})\mid \Jmc'\models E(\vec{t})\}$, with $m=|\{a\mid \Jmc'\models E(\vec{t}), a\in \vec{t} \}|$ and $\linearmap$ being vector concatenation. 

Let $\eta$ be the interpretation of $(\NC,\NR,\dlNames{i},\dlNames{a} )$ such that $\eta(a)=\eta'(a)$ for every $a\in\dlNames{i}$ and for every $E\in \NC\cup\NR$ and ground $S\in \Slang$, if $E_S$ occurs in $\dltrans{\Omc_{\sf g}}$ then $\eta(\syntaxpatom{E}{S})=\eta'(E_S)$, and otherwise $\eta(\syntaxpatom{E}{S})=\emptyset$. 
Since for all $S\implies T$ where $E_S$ and $E_T$ 
occur in $\dltrans{\Omc_{\sf g}}$, $E_S\sqsubseteq E_T\in \dltrans{\Omc_{\sf g}}$, then $\eta'(E_S)\subseteq  \eta'(E_T)$. 
Moreover, if $S$ does not occur in $\dltrans{\Omc_{\sf g}}$, then $\eta(\syntaxpatom{E}{S})=\emptyset$.  
Thus for all $S\implies T$, $\eta(\syntaxpatom{E}{S})\subseteq  \eta(\syntaxpatom{E}{T})$. It follows that $\eta$ is a $m$-dimensional $\linearmap$-geometric interpretation. 
It is clear that $\eta$ is convex since $\eta'$ is convex. 

We show that $\Imc(\eta, \NIkbstar{\Omc})$ is a model of $\Omc$. 
We have the following 
\begin{itemize}
\item $\Delta^{\Imc(\eta, \NIkbstar{\Omc})}_\ind=\mathbb{R}^m$ 
and, for all $a\in\dlNames{i}$, $a^{\Imc(\eta, \NIkbstar{\Omc})}=\eta(a)$;
\item $\Delta^{\Imc(\eta, \NIkbstar{\Omc})}_\ann= \dlNames{a}$ and for all $a\in \dlNames{a}$, $a^{\Imc(\eta,  \NIkbstar{\Omc})}=a$;
\item $A^{\Imc(\eta, \NIkbstar{\Omc})}=\{(\delta,F^\star_S)\mid \delta \in\eta(\syntaxpatom{A}{S}), 
S \text{ built on }\NIkb{\Omc}\}=
\{(\delta,F^\star_S)\mid \delta \in\eta'(A_S), $ \\ $ S \text{ built on }\NIkb{\Omc}\}$, for all $A\in\NC$;
\item  $R^{\Imc(\eta, \NIkbstar{\Omc})}=\{(\delta, \epsilon ,F^\star_S)\mid f(\delta,\epsilon)=\delta\oplus \epsilon \in\eta(\syntaxpatom{R}{S}), S \text{ built on }\NIkb{\Omc}\}=\{(\delta, \epsilon ,F^\star_S)\mid \delta\oplus \epsilon \in\eta'(R_S), S \text{ built on }\NIkb{\Omc}\}$, for all $R\in\NR$.
\end{itemize}

Let $\syntaxpatom{E(\vec{t})}{S}$ be an assertion in $\Omc$. Since $\Jmc'$ is a model of $\dltrans{\Omc_{\sf g}}$, then $\Jmc' \models E_S(\vec{t})$ and $\eta' \models E_S(\vec{t})$. 
Thus $ \oplus_{t\in\vec{t}} \eta(t)\in\eta(\syntaxpatom{E}{S})$, where $ \oplus_{t\in\vec{t}} \eta(t)=t$ if $\vec{t}=(t)$ and $ \oplus_{t\in\vec{t}} \eta(t)=t\oplus u$ if $\vec{t}=(t,u)$. 
By Theorem \ref{lem:fundam}, it follows that ${\Imc(\eta,  \NIkbstar{\Omc})}\models \syntaxpatom{E(\vec{t})}{S}$.
\smallskip

Let $\alpha = \syntaxsetrestriction{X_1}{T_1}, \ldots, \syntaxsetrestriction{X_n}{T_n} \  (\bigsqcap_{i=1}^k \syntaxpatom{E_i}{S_i}\sqsubseteq \syntaxpatom{E}{S})$ be a (concept or role) inclusion in $\Omc$ (note that in case of role inclusion, $k=n=1$ and the inclusion is positive). 
Let  $\Zuweisung$  be a variable assignment mapping object variables $x \in \objectvariables$ to elements $\Zuweisung(x) \in \NIkb{\Omc}$ and set variables  $X\in \setvariables$ to finite binary  relations $\Zuweisung(X) \in \asdom{\NIkb{\Omc}}$ that satisfies $\Zuweisung(X_1)\in T_1^{{\Imc(\eta,  \NIkbstar{\Omc})},\Zuweisung},\dots, \Zuweisung(X_n)\in T_n^{{\Imc(\eta,  \NIkbstar{\Omc})},\Zuweisung}$ 
 and let $\vec{\delta}\in  (\bigsqcap_{i=1}^k \syntaxpatom{E_i}{S_i})^{{\Imc(\eta,  \NIkbstar{\Omc})},\Zuweisung}$.
 \begin{enumerate}
\item For every $1\leq i\leq k$, $\vec{\delta}\in  \syntaxpatom{E_i}{S_i}^{{\Imc(\eta,  \NIkbstar{\Omc})},\Zuweisung}$.   
 \item It follows that for every $1\leq i\leq k$, $(\vec{\delta}, F_i)\in E_i^{\Imc(\eta,  \NIkbstar{\Omc})}$ for some $F_i\in S_i^{{\Imc(\eta,  \NIkbstar{\Omc})},\Zuweisung}$. 
 \item By (2) and by construction of $E_i^{\Imc(\eta,  \NIkbstar{\Omc})}$, $\oplus_{\delta\in\vec{\delta}} \delta\in\eta(\syntaxpatom{E_i}{S_{F_i}})$ where 
 $S_{F_i}$ is a (closed or open) specifier containing exactly the attribute-value pairs that occur in $F_i$, except $(\star,\star)$. 
\item Let ${S_i}_\Zuweisung$ be the specifier obtained by grounding $S_i$ according to $\Zuweisung$. By (2), $F_i\in S_i^{{\Imc(\eta,  \NIkbstar{\Omc})},\Zuweisung}$ so if $S_i$ is closed $S_{F_i}={S_i}_\Zuweisung$, and otherwise, $S_{F_i}\implies {S_i}_\Zuweisung$.
 \item By (4),  $\dltrans{\Omc_{\sf g}}$ contains by construction ${E_i}_{S_{F_i}}\sqsubseteq {E_i}_{{S_i}_\Zuweisung}$ for every $1\leq i\leq k$ such that ${S_i}_\Zuweisung\neq S_{F_i}$, as well as the DL translation of the $\Zuweisung$-instance of $\alpha$ :  $\bigsqcap_{i=1}^k{E_i}_{{S_i}_\Zuweisung}\sqsubseteq {E}_{{S}_\Zuweisung}\in \dltrans{\Omc_{\sf g}}$. 
 \item  By (5) and since $\eta'$ is a geometric model of $\dltrans{\Omc_{\sf g}}$:  
 $$\bigcap_{i=1}^k\eta'({E_i}_{S_{F_i}})\subseteq\bigcap_{i=1}^k \eta'({E_i}_{{S_i}_\Zuweisung}).$$ 
 \item Hence $\eta'(\bigsqcap_{i=1}^k{E_i}_{{S_i}_\Zuweisung})\subseteq \eta'({E}_{{S}_\Zuweisung})$. 
 \item  Thus $\oplus_{\delta\in\vec{\delta}} \delta\in\eta(\bigsqcap_{i=1}^k{E_i}_{S_{F_i}})=\eta'(\bigsqcap_{i=1}^k{E_i}_{S_{F_i}})$ implies that $\oplus_{\delta\in\vec{\delta}} \delta\in\eta'({E}_{{S}_\Zuweisung})=\eta({E}_{{S}_\Zuweisung})$. 
\item If $\alpha$ is a positive inclusion, it follows from (8) that $(\vec{\delta}, F^\star_{{S}_\Zuweisung})\in E^{\Imc(\eta, \NIkbstar{\Omc})}$.  It is then easy to see that $F^\star_{{S}_\Zuweisung}\in S^{{\Imc(\eta,  \NIkbstar{\Omc})},\Zuweisung}$, since $F^\star_{{S}_\Zuweisung}$ is the annotation set that contains the attribute-value pairs that occur
 in the grounding of $S$ w.r.t. $\Zuweisung$, plus $(\star,\star)$ if $S$ is open. Hence $\vec{\delta}\in \syntaxpatom{E}{S}^{\Imc(\eta, \NIkbstar{\Omc})}$. 
 \item If $\alpha$ is a negative concept inclusion, with  $\syntaxpatom{E}{S}=\bot$, it follows from (8) that $\delta\in \emptyset$, which is a contradiction. Hence $(\bigsqcap_{i=1}^k \syntaxpatom{E_i}{S_i})^{{\Imc(\eta,  \NIkbstar{\Omc})},\Zuweisung}=\emptyset$.
 \item By (9) and (10),  $(\bigsqcap_{i=1}^k \syntaxpatom{E_i}{S_i})^{{\Imc(\eta,  \NIkbstar{\Omc})},\Zuweisung}\subseteq \syntaxpatom{E}{S}^{\Imc(\eta, \NIkbstar{\Omc})}$, hence ${\Imc(\eta,  \NIkbstar{\Omc})}\models \alpha$. 
\end{enumerate}
 As $\alpha$ was an arbitrary (concept or role) inclusion in \Omc, 
 it follows that
 ${\Imc(\eta,  \NIkbstar{\Omc})}\models\Omc$, and $\eta$ is a convex geometric model of $\Omc$.
\end{proof}

\section{Proof of Theorem~12}

\temporalgeometric*
\begin{proof}
We first ground the ontology \Omc. 
Let $\akb$ be a satisfiable \dlliteat ontology and let $\NIkb{\akb}$ 
  be the union of the sets of    
  annotation names, time points, and intervals,   
  occurring in $\akb$, respectively.
  Let $\Jmc$ be an interpretation  of \dlliteat over the domain
  $\Delta^{\Jmc} = \NIkb{\akb}\cup \{ \freshconstant \}\cup \{max(\NIkb{\akb})+1,max(min(\NIkb{\akb})-1,0)\}$,
  where $\freshconstant$ is a fresh individual name,
  satisfying $a^{\Jmc} = a$ for all $a \in \NIkb{\akb}$. 
  Let $\Zuweisung: \NV \to \asdom{\Jmc}_\akb$ 
  be a variable assignment, where 
  $\asdom{\Jmc}_\akb  \defeq\powset{\Delta^\Jmc\times \Delta^\Jmc}$. 
Consider a concept or 
role inclusion $\alpha$ of the form 
  $\syntaxsetrestriction{X_1}{S_1}, \ldots, \syntaxsetrestriction{X_n}{S_n} \quad  (\bigsqcap_{i=1}^k K_i\sqsubseteq L)$. 
 A variable assignment $\Zuweisung$ is said to be compatible 
with $\alpha$ if $\Zuweisung(X_1)\in S_1^{\Jmc,\Zuweisung},\dots, \Zuweisung(X_n)\in S_n^{\Jmc,\Zuweisung}$.
The \emph{$\Zuweisung$-instance} $\alpha_\Zuweisung$ of $\alpha$ is the 
concept or role inclusion $\bigsqcap_{i=1}^k K'_i\sqsubseteq L'$ obtained by:
\begin{itemize} 
\item replacing each   $X_i$ with $\closedAnnotation{a \syntaxequality b\mid  \tuple{a,b}\in\Zuweisung(X_i) }$; 
\item replacing each assignment $a \syntaxequality \syntaxintent{X}{b}$ 
occurring in some specifier by all assignments $a \syntaxequality c$ such that $\tuple{b,c}\in\Zuweisung(X_i)$; and, 
\item replacing each object variable $ x$ by $\Zuweisung(x)$.
\end{itemize}
Let $\Omc_{\sf g}$ be the ontology containing all $\Zuweisung$-instances $\alpha_\Zuweisung$ of all 
concept or role inclusions $\alpha$ in $\Omc$ and all compatible variable assignments $\Zuweisung$ (they are a finite 
number since $\Delta^{\Jmc}$ is finite). Recall that since specifiers attached to roles are atemporal, the specifiers attached to roles in $\Omc_{\sf g}$ do not contain any temporal attributes.

Now, we translate $\Omc_{\sf g}$ into an 
\dllitea ontology $\Omc^\dagger_{\sf g}$ as follows.

First, $\Omc^\dagger_{\sf g}$ contains every axiom from $\Omc_{\sf g}$ (but in \dllitea
temporal attribute value pairs will be treated in the same way as the ``regular'' attribute value pairs).
Second, given a ground specifier $S$, we denote by $S(a\syntaxequality b)$
the result of removing all temporal attributes from $S$ and adding the
pair $a\syntaxequality b$.
Let $k_{\sf min}$ and $k_{\sf max}$ be the smallest and the largest numbers occurring in \Omc (or $0$ if none occurs).
Moreover, let $S_\ta$ be the set of temporal attribute-value pairs in $S$.
By assumption, there are no temporal attributes on roles.
Then,  
for each $\syntaxconcept{A}{S}$ 
with $S_\ta\neq\emptyset$, $\Omc^\dagger_{\sf g}$ contains the equivalence 
(as usual, $\equiv$ refers to bidirectional $\sqsubseteq$ here):
\begin{equation}\label{add}
\syntaxconcept{A}{S}\equiv \!\bigsqcap_{(a:b)\in S_\ta} \!( \syntaxconcept{A}{S(a:b)})^\sharp
\end{equation}
where the concept expressions $(\syntaxconcept{A}{S(a:b)})^\sharp$  are defined as follows:
 \begin{itemize}
 \item $(\syntaxconcept{A}{S(\during:v)})^\sharp= \bigsqcap_{u\in
 \intervals{v}} \syntaxconcept{A}{S(\during:u)} $ 
\item $(\syntaxconcept{A}{S(\dt:k)})^\sharp= (\syntaxconcept{A}{S(\during:[k,k])})^\sharp $ 
 \item $( \syntaxconcept{A}{S(\since:k)})^\sharp=(\syntaxconcept{A}{S(\during:[k,k_{\sf max}])})^\sharp\sqcap
 \syntaxconcept{A}{S(\since:k_{\sf max})}$ 
\item $( \syntaxconcept{A}{S(\unt:k)})^\sharp=(\syntaxconcept{A}{S(\during:[k_{\sf min},k])})^\sharp\sqcap
 \syntaxconcept{A}{S(\unt:k_{\sf min})}$  
\end{itemize}
with $k\neq k_{\sf min}$ and $k\neq k_{\sf max}$. For $k\in\{k_{\sf min},k_{\sf max}\}$,  we have $(\syntaxconcept{A}{S(a:k)})^\sharp=\syntaxconcept{A}{S(a:k)}$.

Finally, given attribute-value pairs $a\syntaxequality b$ and $c\syntaxequality d$ 
for temporal attributes $a$ and $b$,
we say that $a\syntaxequality b$ \emph{implies} $c\syntaxequality d$
if $\syntaxpatom{A(e)}{\closedAnnotation{a\syntaxequality b}}\models\syntaxpatom{A(e)}{\closedAnnotation{c\syntaxequality d}}$ 
for some arbitrary $A\in\NC$ and $e\in\NI$. 
We then extend $\Omc^\dagger_{\sf g}$ with all inclusions
$\syntaxconcept{A}{S}\sqsubseteq \syntaxconcept{A}{T}$  and $\syntaxconcept{R}{S}\sqsubseteq \syntaxconcept{R}{T}$,
where $\syntaxconcept{A}{S},\syntaxconcept{A}{T}$ and $\syntaxconcept{R}{S},\syntaxconcept{R}{T}$
  occur  in $\Omc^\dagger_{\sf g}$,
including those introduced in \eqref{add}, such that 
for each temporal attribute-value pair $c\syntaxequality d$ in $T$ there is
a temporal attribute-value pair $a\syntaxequality b$ in $S$ such that
$a\syntaxequality b$ implies $c\syntaxequality d$ and:
\begin{itemize}
\item $T$ is an open specifier and 
the set of non-temporal attribute-value pairs
  in $S$ is a superset
of the set of non-temporal attribute-value pairs in $T$; or
\item $S,T$ are closed specifiers and 
the set of non-temporal attribute-value pairs
 in $S$ is equal to the set of non-temporal attribute-value pairs in
 $T$.
\end{itemize}
Whenever this happens we write $S\Rightarrow T$.

This finishes the construction of $\Omc^\dagger_{\sf g}$. 
By Theorem \ref{dlliteaconvexmodel}, $\Omc^\dagger_{\sf g}$ has a $m$-dimensional convex $f$-model $\eta'$. Let $\eta$ be defined as follows:
\begin{itemize}
\item $\eta(a)=\eta'(a)$ for every individual name $a$, 
\item for every role name $R$ and atemporal specifier $S$, 
$\eta(\syntaxpatom{R}{S})= \eta'(\syntaxpatom{R}{S})$,
\item for every concept name $A$ and specifier $S$, 
$\eta(\syntaxpatom{A}{S})=
\bigcap_{(a:b)\in S_\ta} 
\eta'( \syntaxconcept{A}{S(a:b)})^\sharp$ where
$\eta'( \syntaxconcept{A}{S(a:b)})^\sharp$ is defined as follows:
\begin{itemize}
 \item $\eta'(\syntaxconcept{A}{S(\during:v)})^\sharp= \bigcap_{u\in
 \intervals{v}} \eta'(\syntaxconcept{A}{S(\during:u)}) $ 
\item $\eta'(\syntaxconcept{A}{S(\dt:k)})^\sharp= \eta'(\syntaxconcept{A}{S(\during:[k,k])})^\sharp $ 
 \item $\eta'( \syntaxconcept{A}{S(\since:k)})^\sharp=\eta'(\syntaxconcept{A}{S(\during:[k,k_{\sf max}])})^\sharp\cap
 \eta'(\syntaxconcept{A}{S(\since:k_{\sf max})})$ 
\item $\eta'( \syntaxconcept{A}{S(\unt:k)})^\sharp=\eta'(\syntaxconcept{A}{S(\during:[k_{\sf min},k])})^\sharp\cap
 \eta'(\syntaxconcept{A}{S(\unt:k_{\sf min})})$  
\end{itemize}
with $k\neq k_{\sf min}$ and $k\neq k_{\sf max}$. For $k\in\{k_{\sf min},k_{\sf max}\}$,  we have $\eta'(\syntaxconcept{A}{S(a:k)})^\sharp=\eta'(\syntaxconcept{A}{S(a:k)})$.
\end{itemize}
We have that $\eta$ is convex since the $\eta'(\syntaxpatom{R}{S})$ and $\eta'(\syntaxpatom{A}{S(a:b)})$ are convex, and intersections of convex regions are convex. 

Let $(\eta_j)_{j\in\Delta_T}$ be such that 
\begin{itemize}
\item $\eta_j(a)=\eta(a)=\eta'(a)$ for every individual name $a$ and time point $j$, 
\item for every concept name $A$, specifier $S$ without temporal attribute, 
and time point $k_{\sf min}\leq j \leq k_{\sf max}$, $\eta_j(\syntaxpatom{A}{S})=\eta'(\syntaxpatom{A}{S(\during:v)})$ with $j\in v$.
For $j>k_{\sf max}$, we have $\eta_j(\syntaxpatom{A}{S})=\eta'(\syntaxpatom{A}{S(\during:[k_{\sf max},k_{\sf max}+1])})$
and, for $j<k_{\sf min}$, we have $\eta_j(\syntaxpatom{A}{S})=\eta'(\syntaxpatom{A}{S(\during:[k_{\sf min}-1,k_{\sf min}])})$. 
\end{itemize}

We show that $\eta$ is global for $(\eta_j)_{j\in \Delta_T}$ and $\NIkb{\Omc}$, \ie $\Imc(\eta, \NIkb{\Omc})$ is global for $\Imc(\eta_j, \NIkb{\Omc}\setminus (\NTimepoints\cup\NIntervals))_{j\in \Delta_T}$: let $A\in\NC$, $\delta\in \mathbb{R}^m$ and $F$ be a finite binary relation over $\NIkb{\Omc}$. 
We have $(\delta, F)\in A^{\Imc(\eta, \NIkb{\Omc})}$ iff $\delta\in \eta(\syntaxpatom{A}{S})$ for a specifier $S$ that contains exactly the same attribute-value pairs as $F$, except $(\star,\star)$. 
This is the case exactly when for all $(a,x)\in F$:
\begin{itemize}
\item if $a=\dt$, then $\delta\in \eta'(\syntaxpatom{A}{S(\during : [x,x])})$;

\item if $a= \unt$, then $\delta\in \eta'(\syntaxpatom{A}{S(\during : [j,j])})$ for every $j\leq min(x,k)$;

\item if $a= \since$, then $\delta\in \eta'(\syntaxpatom{A}{S(\during : [j,j])})$ for every $x\leq j\leq k_{\sf max}$;

\item if $a= \during$, then $\delta\in \eta'(\syntaxpatom{A}{S(\during : [j,j])})$ for every $j\in x$ and $j\leq k_{\sf max}$.
\end{itemize}
Let $R\in\NR$, $(\delta_1,\delta_2)\in \mathbb{R}^{2m}$ and $F$ be a finite binary relation over $\NIkb{\Omc}$. 
We have $(\delta_1,\delta_2, F)\in R^{\Imc(\eta, \NIkb{\Omc})}$ iff $f(\delta_1,\delta_2)\in \eta(\syntaxpatom{R}{S})$ for a specifier $S$ that contains exactly the same attribute-value pairs as $F$, except $(\star,\star)$. 
By assumption, $S$ is atemporal so $F$ is built over $\NIkb{\Omc}\setminus(\NTimepoints\cup\NIntervals)$ and $(\delta_1,\delta_2, F)\in R^{\Imc(\eta_j, \NIkb{\Omc})}$ for every $j$. 
Hence we have shown that $\Imc(\eta, \NIkb{\Omc})$ is global for $(\Imc(\eta_j, \NIkb{\Omc}\setminus(\NTimepoints\cup\NIntervals)))_{j\in \Delta_T}$.

Finally, we show below that $\Imc(\eta, \NIkbstar{\Omc})$ satisfies \Omc. 
We write the argument for assertions and concept inclusions and omit for  
 role inclusions because it is similar to the one for concept inclusions, but simpler since role inclusions do not have temporal attributes.

\noindent $\bullet$ Let $\syntaxpatom{A(t)}{S}$ be a concept assertion in $\Omc$. 
Since $\Jmc\models \syntaxpatom{A(t)}{S}$ and interprets every individual name by itself, then $(t, F_S)\in A^\Jmc$ 
(since $S$ is a closed specifier). 
Since $\Jmc$ is global for $(\Jmc_j)_{j\in \Delta_T}$, 
it follows that for every $a:x\in S$, 
\begin{itemize}
\item if $a=\dt$, then $\eta'(t)\in\eta'(\syntaxpatom{A}{S(\during: [x,x])})$
because $\syntaxpatom{A}{S}\sqsubseteq\syntaxpatom{A}{S(\during: [x,x])}\in \Omc^\dagger_{\sf g}$ and $\eta'$ is an $m$-dimensional
convex $f$-model of $\Omc^\dagger_{\sf g}$; 
\item  if $a=\unt$, then, for every $k_{\sf min}\leq j\leq  x$, $\eta'(t)\in\eta'(\syntaxpatom{A}{S(\during: [j,j])})$
and $\eta'(t)\in\eta'(\syntaxpatom{A}{S(\unt: k_{\sf min})})$
because $\syntaxpatom{A}{S}\sqsubseteq\syntaxpatom{A}{S(\during: [j,j])}\in \Omc^\dagger_{\sf g}$, 
for every $k_{\sf min}\leq j\leq  x$, $\syntaxpatom{A}{S}\sqsubseteq\syntaxpatom{A}{S(\unt: k_{\sf min})}\in \Omc^\dagger_{\sf g}$, 
and $\eta'$ is an $m$-dimensional
convex $f$-model of $\Omc^\dagger_{\sf g}$;
\item  if $a=\since$, then, for every $x\leq j\leq k_{\sf max}$, $\eta'(t)\in\eta'(\syntaxpatom{A}{S(\during: [j,j])})$
and $\eta'(t)\in\eta'(\syntaxpatom{A}{S(\since: k_{\sf max})})$
because $\syntaxpatom{A}{S}\sqsubseteq\syntaxpatom{A}{S(\during: [j,j])}\in \Omc^\dagger_{\sf g}$, for every $x\leq j\leq k_{\sf max}$,
$\syntaxpatom{A}{S}\sqsubseteq\syntaxpatom{A}{S(\since: k_{\sf max})}\in \Omc^\dagger_{\sf g}$, 
and $\eta'$ is an $m$-dimensional
convex $f$-model of $\Omc^\dagger_{\sf g}$;

\item if $a= \during$, then for every $j\in x$, 
$\eta'(t)\in\eta'(\syntaxpatom{A}{S(\during: [j,j])})$
because $\syntaxpatom{A}{S}\sqsubseteq \syntaxpatom{A}{S(\during: [j,j])}\in \Omc^\dagger_{\sf g}$. 
\end{itemize}
Since $\eta(\syntaxpatom{A}{S})=
\bigcap_{(a:b)\in S_\ta} 
\eta'( \syntaxconcept{A}{S(a:b)})^\sharp$, 
it follows that $\eta(t)= \eta'(t)$ belongs to $\eta(\syntaxpatom{A}{S})$, \ie $\eta\models \syntaxpatom{A(t)}{S}$. 

\noindent $\bullet$ Let $\syntaxpatom{R(t_1,t_2)}{S}$ be a role assertion in $\Omc$. 
Since $\Jmc\models \syntaxpatom{R(t_1,t_2)}{S}$ and interprets every individual name by itself, 
then $(t_1,t_2, F_S)\in R^\Jmc$ (since $S$ is a closed specifier). Since $\Jmc$ is 
global for $(\Jmc_j)_{j\in \Delta_T}$, and by construction of $\Omc^\dagger_{\sf g}$, it follows that $\syntaxpatom{R(t_1,t_2)}{S}\in\Omc^\dagger_{\sf g}$. 
It follows that $f(\eta(t_1),\eta(t_2))= f(\eta'(t_1),\eta'(t_2))\in \mathbb{R}^{2 m}$ belongs to $\eta(\syntaxpatom{R}{S})=\eta'(\syntaxpatom{R}{S})$, \ie $\eta\models \syntaxpatom{R(t_1,t_2)}{S}$.

\noindent $\bullet$ Let $I = \syntaxsetrestriction{X_1}{S_1}, \ldots, \syntaxsetrestriction{X_n}{S_n} \quad  
(\bigsqcap_{i=1}^n \syntaxpatom{B_i}{T_i}\sqsubseteq \syntaxpatom{B}{S})$ be a concept inclusion in $\Omc$. 
Let  $\Zuweisung$  be a variable assignment mapping object variables $x \in \objectvariables$ to elements 
$\Zuweisung(x) \in \dlNames{a}\cup\NTimepoints\cup\NIntervals$ and set variables 
 $X\in \setvariables$ to finite binary  relations $\Zuweisung(X)$ over $\dlNames{a}\cup\NTimepoints\cup\NIntervals$ 
 that satisfy $\Zuweisung(X_1)\in S_1^{{\Imc(\eta,  \NIkbstar{\Omc})},\Zuweisung},\dots, \Zuweisung(X_n)\in S_n^{{\Imc(\eta,  \NIkbstar{\Omc})},\Zuweisung}$
 and let $\delta\in (\bigsqcap_{i=1}^k \syntaxpatom{B_i}{T_i})^{{\Imc(\eta,  \NIkbstar{\Omc})},\Zuweisung}$.  
For every $1\leq i\leq k$, we have $(\delta, F_i)\in B_i^{\Imc(\eta,  \NIkbstar{\Omc})}$ for some $F_i\in T_i^{{\Imc(\eta,  \NIkbstar{\Omc})},\Zuweisung}$. 
It follows from the construction of $B_i^{\Imc(\eta,  \NIkbstar{\Omc})}$ that $\delta \in\eta(\syntaxpatom{B_i}{S_{F_i}})$ where 
 $S_{F_i}$ is a (closed or open) specifier containing exactly the attribute-value pairs that occur in $F_i$ except $(\star,\star)$. 
 By construction, if $B_i=A_i\in \NC$, then $\eta(\syntaxpatom{B_i}{S_{F_i}})=\bigcap_{(a:b)\in S_{F_i, \ta}} 
\eta'( \syntaxconcept{A_i}{S_{F_i}(a:b)})^\sharp$. 
 Otherwise, if $B_i=\exists P_i$, then $\eta(\syntaxpatom{B_i}{S_{F_i}})=\eta'(\syntaxpatom{\exists P_i}{S_{F_i}})$. 
 Moreover, since $F_i\in T_i^{{\Imc(\eta,  \NIkbstar{\Omc})},\Zuweisung}$, then if $T_i$ is closed 
 $S_{F_i}={T_i}_\Zuweisung$, and otherwise, $S_{F_i}\implies {T_i}_\Zuweisung$, where ${T_i}_\Zuweisung$ is the specifier obtained by grounding $T_i$ according to $\Zuweisung$.  
 In the case $B_i=A_i$,
this implies that for every $j$, $\eta'(\syntaxpatom{A_i}{S_{F_i}(\during : [j,j])})
\subseteq \eta'(\syntaxpatom{A_i}{{T_i}(\during : [j,j])_\Zuweisung})$ by the semantics 
of geometric interpretations. 
 It follows that $\eta(\syntaxpatom{B_i}{S_{F_i}})\subseteq\bigcap_{(a:b)\in {T_i}_{\Zuweisung, \ta}} 
\eta'( \syntaxconcept{A_i}{{T_i}_\Zuweisung(a:b)})^\sharp$. 
In the case $B_i=\exists P_i$, this implies that $\eta'(\syntaxpatom{P_i}{S_{F_i}})\subseteq \eta'(\syntaxpatom{P_i}{{T_i}_\Zuweisung})$, so $\delta\in \eta(\syntaxpatom{\exists P_i}{S_{F_i}})=\eta'(\syntaxpatom{\exists P_i}{S_{F_i}})\subseteq \eta'(\syntaxpatom{\exists P_i}{{T_i}_\Zuweisung})$. 

Since $\eta'$ is a model of $\Omc^\dagger_{\sf g}$ which contains 
a $\Zuweisung$-instance $\bigsqcap_{i=1}^n \syntaxpatom{B_i}{T_{i \Zuweisung}}\sqsubseteq 
\syntaxpatom{B}{S_\Zuweisung}$ of the concept inclusion $I$ and equivalences
\begin{equation}\label{add}
\syntaxconcept{B_{(i)}}{S_\Zuweisung}\equiv \!\bigsqcap_{(a:b)\in S_{\Zuweisung\ta}} \!( \syntaxconcept{B_{(i)}}{S_\Zuweisung(a:b)})^\sharp 
\end{equation}
for each $B_{(i)}$. 
This implies that $\delta\in \eta'(\syntaxpatom{B}{{S_\Zuweisung(\during : [j,j])}})$ for every $j$ such that 
$S_\Zuweisung \Rightarrow S_\Zuweisung(\during : [j,j])$. 
By construction of 
$\eta(\syntaxpatom{B}{{S}_\Zuweisung})$, it follows that $\delta\in \eta(\syntaxpatom{B}{{S}_\Zuweisung})$. 
It follows that $(\delta, F^\star_{{S}_\Zuweisung})\in B^{\Imc(\eta, \NIkbstar{\Omc})}$. We can see that $F^\star_{{S}_\Zuweisung}\in S^{{\Imc(\eta,  \NIkbstar{\Omc})},\Zuweisung}$, since $F^\star_{{S}_\Zuweisung}$ is the annotation set that contains the attribute-value pairs that occur
 in the grounding of $S$ w.r.t. $\Zuweisung$, plus $(\star,\star)$ if the specifier is open. 
 So $\delta\in (\syntaxpatom{B}{{S}})^{\Imc(\eta, \NIkbstar{\Omc}),\Zuweisung}$. 
This means that $\eta$ satisfies $I$. 

\noindent $\bullet$ The argument for $I$ of the form 
$\syntaxsetrestriction{X_1}{S_1}, \ldots, \syntaxsetrestriction{X_n}{S_n} \quad 
 (\bigsqcap_{i=1}^n \syntaxpatom{B_i}{T_i}\sqsubseteq \bot)$ is in the same lines
 but we have to argue that 
 if there is a variable assignment $\Zuweisung$ such that
 $(\bigsqcap_{i=1}^n \syntaxpatom{B_i}{T_{i \Zuweisung}})^{\Imc(\eta, \NIkbstar{\Omc})}$ is not the empty set
 then there is also a variable assignment $\Zuweisung'$
 such that 
  $(\bigsqcap_{i=1}^n \syntaxpatom{B_i}{T_{i \Zuweisung'}})^{\Imc(\eta', \NIkbstar{\Omc})}$ is not the empty set
 which contradicts the assumption that 
 $\eta'$
 is a convex $m$-dimensional model of  $\Omc^\dagger_{\sf g}$.
\end{proof}

\end{document}

%% file: macros.tex
\usepackage{xspace}
\usepackage{rotating}

\let\phi\varphi

\newcommand{\tuple}[1]{({#1})}
\newcommand{\defeq}{\coloneqq}

\renewcommand{\vec}[1]{\boldsymbol{#1}}
\newcommand{\powset}[1]{\mathcal{P_{\text{\sf{fin}}}}\left(#1\right)}

\newcommand{\lang}[1]{\ensuremath{\mathbf{#1}}} 
\newcommand{\Slang}{\lang{S}} 

\newcommand{\Inter}{\mathcal{I}} 
\newcommand{\Zuweisung}{\mathcal{Z}} 
\newcommand{\asdom}[1]{\Phi^{#1}} 

\newcommand{\syntaxequality}{\kern0.6pt{:}\,}

\newcommand{\syntaxpatom}[2]{#1 @ #2} 
\newcommand{\syntaxconcept}[2]{\syntaxpatom{#1}{#2}}

\newcommand{\syntaxsetrestriction}[2]{{#1}\,{\boldsymbol{:}}\,{#2}}

\newcommand{\syntaxintent}[2]{#1.#2}

\newcommand{\akb}{\ensuremath{\Omc}\xspace} 

\newcommand{\dlname}[1]{\ensuremath{\mathcal{#1}}\xspace}

\newcommand{\EL}{\dlname{EL}}

\newcommand{\dlNames}[1]{\ensuremath{\mathsf{N_{#1}}}\xspace}
\newcommand{\NC}{\dlNames{C}}
\newcommand{\NR}{\dlNames{R}}
\newcommand{\NI}{\dlNames{I}}

\newcommand{\NTimepoints}{\dlNames{T}}
\newcommand{\NIntervals}{\ensuremath{\mathsf{N^2_{T}}}\xspace}
\newcommand{\NIkb}[1]{\ensuremath{\mathsf{\dlNames{}_{\mathnormal{#1}}}}\xspace}
\newcommand{\NIkbstar}[1]{\ensuremath{\mathsf{\dlNames{}^\star_{\mathnormal{#1}}}}\xspace}

\newcommand{\NV}{\dlNames{V}}
\newcommand{\NU}{\dlNames{U}}
\newcommand{\freshconstant}{\ensuremath{{\sf f}}\xspace}

\newcommand{\ie}{i.e.,\xspace}

\mathtoolsset{centercolon}

\providecommand\given{}
\makeatletter
\newcommand\Set@Symbol[1][]{%
    \nonscript\:#1\vert{}
    \allowbreak{}
    \nonscript\:
    \mathopen{}}
\DeclarePairedDelimiterX{\set}[1]{\lbrace}{\rbrace}{%
    \renewcommand{\given}{\Set@Symbol[\delimsize]}
    #1
}

\DeclarePairedDelimiterX{\openAnnotation}[1]{\lfloor}{\rfloor}{\renewcommand{\given}{\Set@Symbol[\delimsize]}#1}
\DeclarePairedDelimiterX{\closedAnnotation}[1]{[}{]}{\renewcommand{\given}{\Set@Symbol[\delimsize]}#1}
\makeatother

\newcommand{\Hmc}{\ensuremath{\mathcal{H}}\xspace}
\newcommand{\Imc}{\ensuremath{\mathcal{I}}\xspace}
\newcommand{\Jmc}{\ensuremath{\mathcal{J}}\xspace}
\newcommand{\Kmc}{\ensuremath{\mathcal{K}}\xspace}

\newcommand{\Mmc}{\ensuremath{\mathcal{M}}\xspace}

\newcommand{\Omc}{\ensuremath{\mathcal{O}}\xspace}

\newcommand{\Rmc}{\ensuremath{\mathcal{R}}\xspace}

\newcommand{\Zmc}{\ensuremath{\mathcal{Z}}\xspace}

\newcommand{\mypar}[1]{\medskip\noindent\textbf{#1~~}}

\newcommand{\mn}[1]{\ensuremath{{\sf #1}}}